\documentclass[sn-mathphys]{sn-jnl}

\jyear{2021}%

\theoremstyle{thmstyleone}%
\newtheorem{theorem}{Theorem}
\newtheorem{proposition}[theorem]{Proposition}%

\theoremstyle{thmstyletwo}%

\theoremstyle{thmstylethree}%
\newtheorem{definition}{Definition}%

\begin{document}

\title[Channel Polarization of 2D-input Quantum Symmetric Channels]{Channel Polarization of Two-dimensional-input Quantum Symmetric Channels}


\author[1]{\fnm{Zhengzhong} \sur{Yi}}\email{zhengzhongyi@cs.hitsz.edu.cn}

\author[1]{\fnm{Zhipeng} \sur{Liang}}\email{liangzhipenghitsz@163.com}

\author*[1]{\fnm{Xuan} \sur{Wang}}\email{wangxuan@cs.hitsz.edu.cn}

\affil[1]{\orgdiv{Harbin Institute of Technology, Shenzhen}. \orgaddress{\city{Shenzhen}, \postcode{518055}, \country{China}}}




\abstract{Being attracted by the property of classical polar code, researchers are trying to find its analogue in quantum fields, which is called quantum polar code. The first step and the key to design quantum polar code is to find out for the quantity which can measure the quality of quantum channels, whether there is a polarization phenomenon which is similar to classical channel polarization. Coherent information is believed to be the quantum analogue of classical mutual information and the quantity to measure the capacity of quantum channel. In this paper, we define a class of quantum channels called quantum symmetric channels, and prove that for quantum symmetric channels, under the similar channel combining and splitting process as in the classical channel polarization, the maximum single letter coherent information of the coordinate channels will polarize. That is to say, there is a channel polarization phenomenon in quantum symmetric channels.}




\maketitle

\section{Introduction}

The potential to solve different problems more efficiently than the state-of-the-art classical computing makes quantum computing attract worldwide attention. To give full play to this potential, quantum computers should have sufficient reliable qubits. However, at present, physical qubits are quite vulnerable, which restricts the development of large-scale fault-tolerant quantum computing and exploiting the advantages of quantum computing. Fortunately, quantum error correcting codes (QECCs) discovered by Shor and Steane provide us with a solution to this problem\cite{PhysRevA.52.R2493,PhysRevLett.77.793}.

Similar to classical error correcting codes (CECCs), QECCs encoding $n$ (which is called \textbf{code length}) less reliable physical qubits (with error rate $p_{0}$) in a certain way to obtain $k$ ($k<n$) more reliable logic qubits (with error rate $p_{L}<p_{0}$ after decoding and recovery). The ratio $k/n$ is called coding rate. The higher it is, the more efficient the QECC is. No matter for CECCs or QECCs, to improve the reliability of the logic bits/qubits, we often need to increase the code length. Good CECCs have constant or increasing coding rate with code length increasing, some\cite{1057683,748992,mackay1996near,5075875} can even asymptotically achieve the channel capacity which is a quantity measures the upper limit of coding rate. However, for most QECC schemes, the larger the code length $n$ is, the lower the coding rate will be, which will results in excessive physical qubits overhead. This makes reliable large-scale fault-tolerant quantum computing needs millions of physical qubits, which is very difficult to realize for the current technology. For Surface Code\cite{bravyi1998quantum,PhysRevA.89.022321,bullock2007qudit,zemor2009cayley,wang2009threshold,PhysRevA.80.052312,bravyi2012subsystem,ghosh2012surface,PhysRevLett.109.180502,PhysRevLett.109.160503,PhysRevA.86.042313,PhysRevA.86.032324,fowler2013optimal,barends2014superconducting,hill2015surface,delfosse2016linear,PhysRevApplied.8.034021,huang2020alibaba} which is the most promising QECC at present, and the concatenated QECCs\cite{doi:10.1137/S0097539799359385,knill1996concatenated,knill2005quantum} which is the earliest and also a promising method to realize fault-tolerant quantum computing, their coding rate tends to 0 with the increase of its code length. For quantum low-density parity check (QLDPC) codes\cite{gottesman2013fault,tillich2013quantum,freedman2013quantum,guth2014quantum,kovalev2013fault,hastings2013decoding,breuckmann2016constructions,breuckmann2017hyperbolic,breuckmann2021single,grospellier2021combining}, though their coding rate is constant with code length increasing, whether their coding rate can achieve the channel capacity has not been proven. In some cases,  such as hyperbolic codes\cite{hastings2013decoding,breuckmann2016constructions,breuckmann2017hyperbolic,grospellier2021combining,breuckmann2021single}, which is a family of QLDPC codes, have a constant coding rate, but their coding rate does not seem to achieve the quantum channel capacity (we measure this capacity by maximum single letter coherent information of the quantum channel, which is explained in Subsection \ref{2.4}). For instance, in ref \cite{breuckmann2021single}, the  asymptotic coding rate of 4D-hyperbolic code is 0.18, but the  quantum channel capacity of the independent $X/Z$-flip noise channel considered by the authors with error rate $p=0.04$ (i.e. a qubit undergoes independently an $X$ error with probability $p$ or a $Z$ error with probability $p$) is 0.5178, which is rather larger than its asymptotic coding rate.

Classical polar code (CPC) is the only error correcting code whose coding rate has been proven that it can reach the classical channel capacity\cite{5075875}. The high coding rate has attracted researchers' attention. In the past decade, researchers are trying to apply the channel polarization idea of CPC to quantum channels and find the analogue of CPC in quantum fields, which is called quantum polar code (QPC)\cite{guo2013polar,renes2012efficient,wilde2013polar,hirche2015polar,7208851,7370934,8989387,8815775}.

The first step and also the key to design QPC is to figure out whether quantum channels will polarize which is similar to classical channel polarization discovered in \cite{5075875}. Some previous studies\cite{6302198,guo2013polar,renes2012efficient,wilde2013polar,hirche2015polar,7208851,7370934,8815775,6284203,goswami2021quantum} have proved some quantities of classical-quantum channels whose inputs are classical bits and outputs are qubits, such as classical symmetric capacity\cite{guo2013polar}, Bhattacharyya parameter\cite{guo2013polar,goswami2021quantum}, and classical symmetric Holevo information\cite{6302198} will polarize. Some studies\cite{wilde2013polar,renes2012efficient,6284203,7208851,goswami2021quantum} has referred to coherent information, which is a quantum quantity of quantum channels and is believed to be a quantity to measure the channel capacity of pure quantum channels\cite{9241807,8242350,holevo2020quantum,5592851,holevo2010mutual,bennett2004quantum,PhysRevA.57.4153,PhysRevA.55.1613,kretschmann2004tema,shor2003capacities}, but the coherent information of the classical-quantum channels is just the classical mutual information. Based on the polarization of classical-quantum channels, researchers have proposed some quantum polar coding schemes\cite{6302198,guo2013polar,renes2012efficient,wilde2013polar,hirche2015polar,7208851,7370934,8815775,6284203,goswami2021quantum}. Unfortunately, they cannot be applied to quantum computing whose quantum channels are pure quantum channels. In 2019, Dupuis\cite{8989387} prove that the symmetric coherent information(the coherent information of quantum channel evaluated for Bell-state input\cite{renes2012efficient}) of pure quantum channels will polarize. However, unlike the classical symmetric capacity having been proved that it is the channel capacity of classical symmetric channels, no one has proved that the symmetric coherent information is the maximum single letter coherent information of pure quantum channels.

In this paper, we focus on proving the polarization of pure quantum channels. We first define a class of quantum channels called quantum symmetric channels (QSCs, this term has been used in \cite{javidian2021quantum}, but in this paper, it has different meaning), and prove some basic properties of them. For QSC, we prove that its maximum single letter coherent information (MSLCI) equals to its symmetric coherent information. Then we prove the MSLCI will polarize in two-dimensional-input QSC under the quantum channel combining and splitting process. Unlike the proof method used by Dupuis\cite{8989387}, our proof uses the basis transition probability matrix proposed by us.

The rest of this paper is organized as follows.  Some preliminary knowledge, including coherent information, quantum symmetric channels, quantum channel combing and splitting, will be introduced in Sect. \ref{2}. In Sect. \ref{3}, we will prove that the combined channel is quantum symmetric channel and all the coordinate channels are two-dimensional-input quasi quantum symmetric channels. In Sect. \ref{4}, we will prove the MSLCI of the coordinate channels will polarize. In Sect. \ref{5}, we conclude our work.

\section {Preliminaries}
\label{2}
\subsection{Coherent information}
\label{2.1}
	Coherent information is proposed by Schumacher and Nielsen  to measure the amount of quantum information conveyed in the noisy channel\cite{PhysRevA.54.2629}. It is believed to be the analogue  of classical mutual information in quantum information theory\cite{nielsen2002quantum}.

\begin{figure}[!t]
\centering
\includegraphics[width=0.85\textwidth]{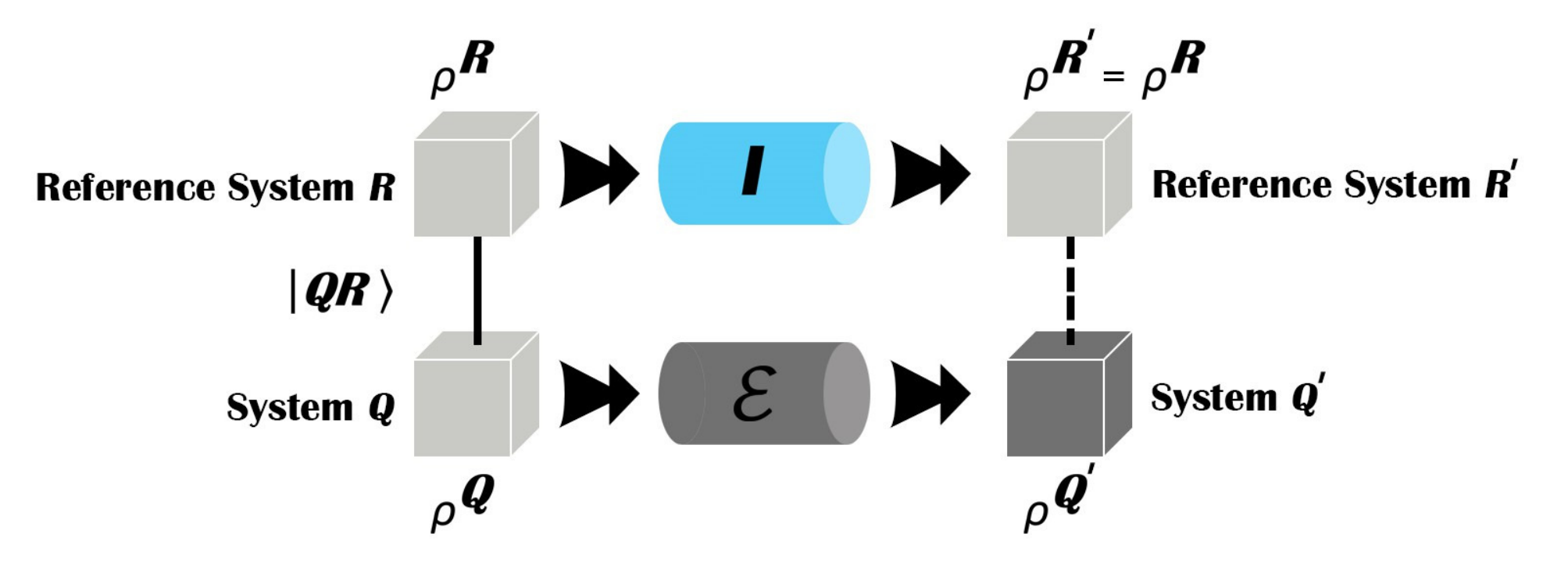}
\caption{System $Q$ and its reference system $R$. System $Q$ is subjected to a quantum channel $\mathcal{E}$. Notice that the reference system $R$ is only subjected to an identity operator $I$, namely, $R^{'}=R$. The solid line between system $Q$ and its reference system $R$ indicates $Q$ and $R$ are in a maximally entangled state, which means in a certain basis, the measurement results of $Q$ and $R$ have an one-to-one relationship. Once you get the measurement results of $Q$, you know the state of $R$, and vice versa. Hence, we use solid line to represent this “strong” relationship. The dashed line indicates there might be still some entanglement between $Q^{'}$ and $R^{'}$ and the one-to-one relationship might not exist.}
\label{fig:Quantum_channel}
\end{figure}
As shown in Fig. \ref{fig:Quantum_channel}, suppose the state of a quantum system $Q$ is $\rho^{Q}$,
\begin{equation}
\rho^{Q}= \sum_{i}p_{i}\ket{i^{Q}}\bra{i^{Q}}
\label{eq:rhoQ}
\end{equation}
where $\ket{i^{Q}}$ is the basis for $Q$. Suppose $Q$ is subjected to a quantum channel $\mathcal{E}$ which change system $Q$ to $Q^{'}$ and maps the state to $\rho^{Q^{'}}$, namely,
\begin{equation}
\rho^{Q^{'}}= \mathcal{E}(\rho^{Q})
\label{eq:rhoQ'}
\end{equation}

For system $Q$, we can always introduce a reference system $R$ which has the same state space as $Q$ to purify $Q$, namely, map the mixed state $\rho^{Q}$ to a pure state $\ket{QR}$. The state of system $Q$ and $R$ can be expressed as
\begin{equation}
\ket{QR}=\sum_{i} \sqrt{p_{i}}\ket{i^{Q}}\ket{i^{R}}
\label{eq:pure state QR}
\end{equation}
where $\ket{i^{R}}$ is the basis for $R$, which is the same as $\ket{i^{Q}}$.
	
Schumacher  defined an intrinsic quantity to $Q$ called entropy exchange $S_{e}$\cite{PhysRevA.54.2614},
\begin{equation}
S_{e}\equiv S(RQ^{'})
\label{eq:entropy exchange}
\end{equation}
where $S(RQ^{'})$ is the von Neumann entropy of system $RQ^{'}$.

Coherent information in the process shown in Fig. \ref{fig:Quantum_channel} is defined as
\begin{equation}
I(Q;Q^{'})\equiv S(Q^{'})-S_{e}=S(Q^{'})-S(RQ^{'})
\label{eq:coherent information}
\end{equation}
where $S(Q^{'})$ is the von Neumann entropy of system $Q^{'}$. It's obvious that once $Q$ and $\mathcal{E}$ are given, $Q^{'}$ is determined. Hence, we can also write $I(Q;Q^{'})$ as $I(\rho^{Q},\mathcal{E})$.

Assuming the operation elements of $\mathcal{E}$ are $\{E_{k}\}$, then $S_{e}$ can be calculated by
\begin{equation}
S_{e}=S(W)
\label{eq:Se=S(W)}
\end{equation}
where $W_{ij}=tr\left(E_{i}\rho^{Q}E_{j}^\dagger\right)$.

It should be emphasized that in this paper, the coherent information which we consider is the single letter coherent information (SLCI). Due to the superadditivity\cite{hastings2009superadditivity} of quantum channel, single letter coherent information is the lower bound of quantum channel capacity. Researchers\cite{cubitt2015unbounded,smith2008quantum} believes the quantum channel capacity should be more accurately measured by $I\left(\rho^{Q}, \mathcal{E}^{\otimes n}\right)$ which is defined by 
\begin{equation}
I\left(\rho^{Q}, \mathcal{E}^{\otimes n}\right) \equiv \lim _{n \rightarrow \infty} \frac{1}{n} I\left(\rho^{Q}, \mathcal{E}\right)
\end{equation}
Whether will $I\left(\rho^{Q}, \mathcal{E}^{\otimes n}\right)$ of the coordinate channels polarize has not been proven in this paper.

\subsection{Quantum symmetric channels}
\label{2.2}
	In classical information theory, there is a class of channels called classical symmetrical channels (CSCs) whose properties have been well-studied, such as binary symmetric channel (BSC). The behavior of a classical channel can be depicted by a transition probability matrix (TPM). Assume the input variable is $A$, which takes value from $\{a_{1},a_{2},\cdots,a_{K}\}$, and the output variable is $B$, which takes value from $\{b_{1},b_{2},\cdots,b_{L}\}$, then we can write out its TRM as follows.

\begin{equation}
\bordermatrix{
		& B=b_{1}             & \cdots & B=b_{L}           \cr
A=a_{1} & p(B=b_{1}\lvert A=a_{1})  & \cdots &p(B=b_{L}\lvert A=a_{1}) \cr
A=a_{2} & p(B=b_{1}\lvert A=a_{2})  & \cdots &p(B=b_{L}\lvert A=a_{2}) \cr
\vdots  &\vdots               & \ddots &\vdots             \cr
A=a_{K} & p(B=b_{1}\lvert A=a_{K})  & \cdots &p(B=b_{L}\lvert A=a_{K})
}
\label{eq:transition probability matrix}
\end{equation}

	If each row of the TPM is a permutation of the first row, then this channel is symmetric with respect to its input. If each column of the TPM is a permutation of the first column, then this channel is symmetric with respect to its output. If a channel is symmetric with respect to both of its input and output, then this channel is called a symmetric channel. If a channel is symmetric with respect to its input but might not to its output, and its TPM can be divided into several submatrices by column, each of which satisfies that each column of it is a permutation of the first column of it, then this channel is called a quasi symmetric channel. 

For some quantum channels, given certain basis of the input space and the output space, we may also find a probability matrix similar to TRM of classical channels. For example, for bit flip channel, if the input state $\ket{Q}$ is $\ket{0}$ (with probability $q$) or $\ket{1}$ (with probability $1-q$), then the output state $\ket{Q^{'}}$ will also take value from $\ket{0}$ or $\ket{1}$, and we can figure out $p(\ket{Q^{'}}=\ket{0}\lvert \ket{Q}=\ket{0})$, $p(\ket{Q^{'}}=\ket{1}\lvert \ket{Q}=\ket{0})$, $p(\ket{Q^{'}}=\ket{0}\lvert \ket{Q}=\ket{1})$, $p(\ket{Q^{'}}=\ket{1}\lvert \ket{Q}=\ket{1})$. Then we can write out a probability matrix as follows.

\begin{equation}
\bordermatrix{
		        & \ket{Q^{'}}=\ket{0}                    & \ket{Q^{'}}=\ket{1}                    \cr
\ket{Q}=\ket{0} & p(\ket{Q^{'}}=\ket{0}\lvert \ket{Q}=\ket{0}) & p(\ket{Q^{'}}=\ket{1}\lvert \ket{Q}=\ket{0}) \cr
\ket{Q}=\ket{1} & p(\ket{Q^{'}}=\ket{0}\lvert \ket{Q}=\ket{1}) & p(\ket{Q^{'}}=\ket{1}\lvert \ket{Q}=\ket{1})
}
\label{eq:transition probability matrix of bit flip channel}
\end{equation}

Here, we name matrix (\ref{eq:transition probability matrix of bit flip channel}) basis transition probability matrix (BTPM), for it shows the transition relationship between the basis of input and output spaces. Different from TPM of classical channels, the above BTPM doesn’t seem to fully depicted the behavior of bit flip channel, because quantum mechanics allow the input state to be a superposition state, such as $\frac{1}{\sqrt{2}}\ket{0}+\frac{1}{\sqrt{2}}\ket{1}$. That is to say, the input state and the output state may take value outside $\{\ket{0},\ket{1}\}$, which cannot be depicted by BTPM. However, using BTPM (\ref{eq:transition probability matrix of bit flip channel}), given arbitrary input state, we can always determine the output. This is because we can write out the operator elements by matrix (\ref{eq:transition probability matrix of bit flip channel}), which will be proved later. And once the operator elements are determined, the behavior of the quantum channel is determined. Hence, for the quantum channels which have a BTPM, its behavior can be fully depicted by its BTPM. 

There is a necessary and sufficient condition for a quantum channel having a BTPM.

\begin{theorem}[\textbf{Necessary and sufficient condition for a quantum channel having a BTPM}]
\label{theorem:Necessary and sufficient condition for a quantum channel having a BTPM}
Given a quantum channel $\mathcal{E}$, it has a BTPM if and only if there is a certain basis $B_{in}=\{\ket{1},\ket{2},\cdots,\ket{N}\}$ of the input space, any two basis vectors $\ket{i}$ and $\ket{j}$ in $B_{in}$ satisfy that $\mathcal{E}(\ket{i}\bra{i})$ commutes with $\mathcal{E}(\ket{j}\bra{j})$, namely,
\begin{equation}
\begin{aligned}
&\left[\mathcal{E}(\ket{i}\bra{i}),\mathcal{E}(\ket{j}\bra{j})\right]=\mathcal{E}(\ket{i}\bra{i})\mathcal{E}(\ket{j}\bra{j})-\mathcal{E}(\ket{j}\bra{j})\mathcal{E}(\ket{i}\bra{i})=0
\end{aligned}
\end{equation}
\end{theorem}

\begin{proof}

\noindent\textbf{(1)Sufficiency:}
If there is a certain basis $B_{in}=\{\ket{1},\ket{2},\cdots,\ket{N}\}$ of the input space, any two basis vectors $\ket{i}$ and $\ket{j}$ in $B_{in}$ satisfy $[\mathcal{E}(\ket{i}\bra{i}),\mathcal{E}(\ket{j}\bra{j})]=0$, then for all $\mathcal{E}(\ket{i}\bra{i})$, they can be simultaneously diagonalized in a certain basis $B_{out}=\{\ket{1^{'}},\ket{2^{'}},\cdots,\ket{M^{'}}\}$ of the output space. The result of diagonalization is
\begin{equation}
\mathcal{E}(\ket{i}\bra{i})=\sum_{k=1}^{M}p_{ik}\ket{k^{'}}\bra{k^{'}}
\end{equation}
It is obvious that $p_{ik}$ forms the BPTM.

\noindent\textbf{(2)Necessity:}
Assume quantum channel $\mathcal{E}$ has a BTPM whose elements are $A_{ik}\ (1\leq i\leq N, 1\leq k\leq M)$, and the corresponding basis for the input and output space are $B_{in}=\{\ket{1},\ket{2},\cdots,\ket{N}\}$ and $B_{out}=\{\ket{1^{'}},\ket{2^{'}},\cdots,\ket{M^{'}}\}$, respectively, then $\mathcal{E}(\ket{i}\bra{i})$ can be expressed as
\begin{equation}
\mathcal{E}(\ket{i}\bra{i})=\sum_{k=1}^{M}A_{ik}\ket{k^{'}}\bra{k^{'}}
\end{equation}
which means that all $\mathcal{E}(\ket{i}\bra{i})$ can be simultaneously diagonalized in $B_{out}=\{\ket{1^{'}},\ket{2^{'}},\cdots,\ket{M^{'}}\}$. Hence, any two basis vectors $\ket{i}$ and $\ket{j}$ in $B_{in}$ satisfy $[\mathcal{E}(\ket{i}\bra{i}),\mathcal{E}(\ket{j}\bra{j})]=0$. The proof is completed. 
\end{proof}

Next, we are going to prove that one can derive the channel operation elements by BTPM.
\begin{theorem}[\textbf{Derive the channel operation elements from BTPM}]
\label{theorem2-2:Derive the channel operation elements by BTPM}
For a quantum channel which has BTPM, its BTPM determine a set of quantum operations.
\end{theorem}

\begin{proof}
Assume quantum channel $\mathcal{E}$ has a BTPM $A$, for arbitrary input $\rho=\sum_{i=1}^{N}q_{i}\ket{i}\bra{i}$, the corresponding output $\mathcal{E}(\rho)$ is
\begin{equation}
\begin{aligned}
\mathcal{E}(\rho)&=\sum_{i=1}^{N}q_{i}\mathcal{E}(\ket{i}\bra{i})=\sum_{i=1}^{N}q_{i}\sum_{k=1}^{M}A_{ik}\ket{k^{'}}\bra{k^{'}}=\sum_{k=1}^{M}E_{k}\left(\sum_{i=1}^{N}q_{i}\ket{i}\bra{i}\right)E_{k}^{\dagger}
\end{aligned}
\end{equation}
where $\{E_{k}\}$ are the operation elements of channel $\mathcal{E}$, and $E_{k}\ket{i}=\sqrt{A_{ik}}\ket{k^{'}}$, according to which one can easily write out the matrix representation of $E_{k}$. This completes the proof.
\end{proof}

According to the above proof, one can see that the number of independent operation elements equals to the number of dimensions of the output space.

Similar to classical symmetric channels, we can define quantum symmetric channels by BTPM.

\begin{definition}[\textbf{Quantum symmetric channels}]
\label{Quantum symmetric channels}
	For the quantum channels which have BTPM, if each row of the BTPM is a permutation of the first row, then this quantum channel is symmetric with respect to its input. If each column of the BTPM is a permutation of the first column, then this quantum channel is symmetric with respect to its output. If a quantum channel is symmetric with respect to both of its input and output, then this channel is called quantum symmetric channel (QSC). If a channel is symmetric with respect to its input but might not to its output, and its BTPM can be divided into several submatrices by column, each of which satisfies that each column of it is a permutation of the first column of it, then this channel is called a quantum quasi symmetric channel (QQSC). Actually, a QSC can be regarded as a special QQSC.
\end{definition}

\begin{theorem}[\textbf{Operation elements of two-dimensional-input QQSC}]
\label{theorem2-3:Operation elements of two-dimensional-input QQSC}
For a two-dimensional-input QQSC whose output space is $M$-dimensional, there is always a set of operation elements $\{E_{k}\}$ , $1\leq k\leq M$, which satisfies
\begin{equation}
\label{Ek}
E_{k}\ket{0}=\sqrt p_{k}\ket{k^{'}},E_{k}\ket{1}=\sqrt p_{k}\ket{\pi(k)^{'}}
\end{equation}
where $\pi$ is a certain permutation, $\{\ket{k^{'}}\}$ is a basis of the output space, $\sum_{k=1}^{M} p_{k}=1$. Notice that $\ket{0}$ and  $\ket{1}$ are only basis vectors of the input space, they are not necessary to be the computational basis vectors $\left(\begin{array}{l}1 \\ 0\end{array}\right)$ and $\left(\begin{array}{l}0 \\ 1\end{array}\right)$.
\end{theorem}

\begin{proof}
We only need to determine the matrix representation of each $E_{k}$ and prove that $\sum_{k} E_{k}^{\dagger} E_{k}=I$. Notice that the matrix representation of $E_{k}$ can be calculated by 

\begin{equation}
\label{Ekij}
E_{k_{j i}}=\bra{j}E_{k}\ket{i-1}
\end{equation}

Here, the ket vector is $\ket{i-1}$ rather than $\ket{i}$. This is because we use $\ket{0}$ and $\ket{1}$ to represent the input basis rather than $\ket{1}$ and $\ket{2}$, so the index of the column of the matrix starts from 1. Through Eq. (\ref{Ek}) and Eq. (\ref{Ekij}), it’s easy to obtain
\begin{equation}
E_{k_{j 1}}=\sqrt{p_{k}} \delta_{j k}
\end{equation}
\begin{equation}
E_{k_{j 2}}=\sqrt{p_{k}} \delta_{j \pi(k)}
\end{equation}
where $\delta$ is the Kronecker Delta. Hence,
\begin{equation}
\sum_{k} E_{k}^{\dagger} E_{k}=\left(\begin{array}{cc}
\sum_{k=1}^{M} p_{k} & 0 \\
0 & \sum_{k=1}^{M} p_{k}
\end{array}\right)=I
\end{equation}
which completes the proof.
\end{proof}

\subsection{Two examples of QSC}
\label{2.3}
Bit flip channel and phase flip channel are two typical QSCs, as shown in Fig. \ref{fig:Bit_flip_channel} and Fig. \ref{fig:Phase_flip_channel}, respectively.

\begin{figure}[!t]
\centering
\includegraphics[width=0.5\textwidth]{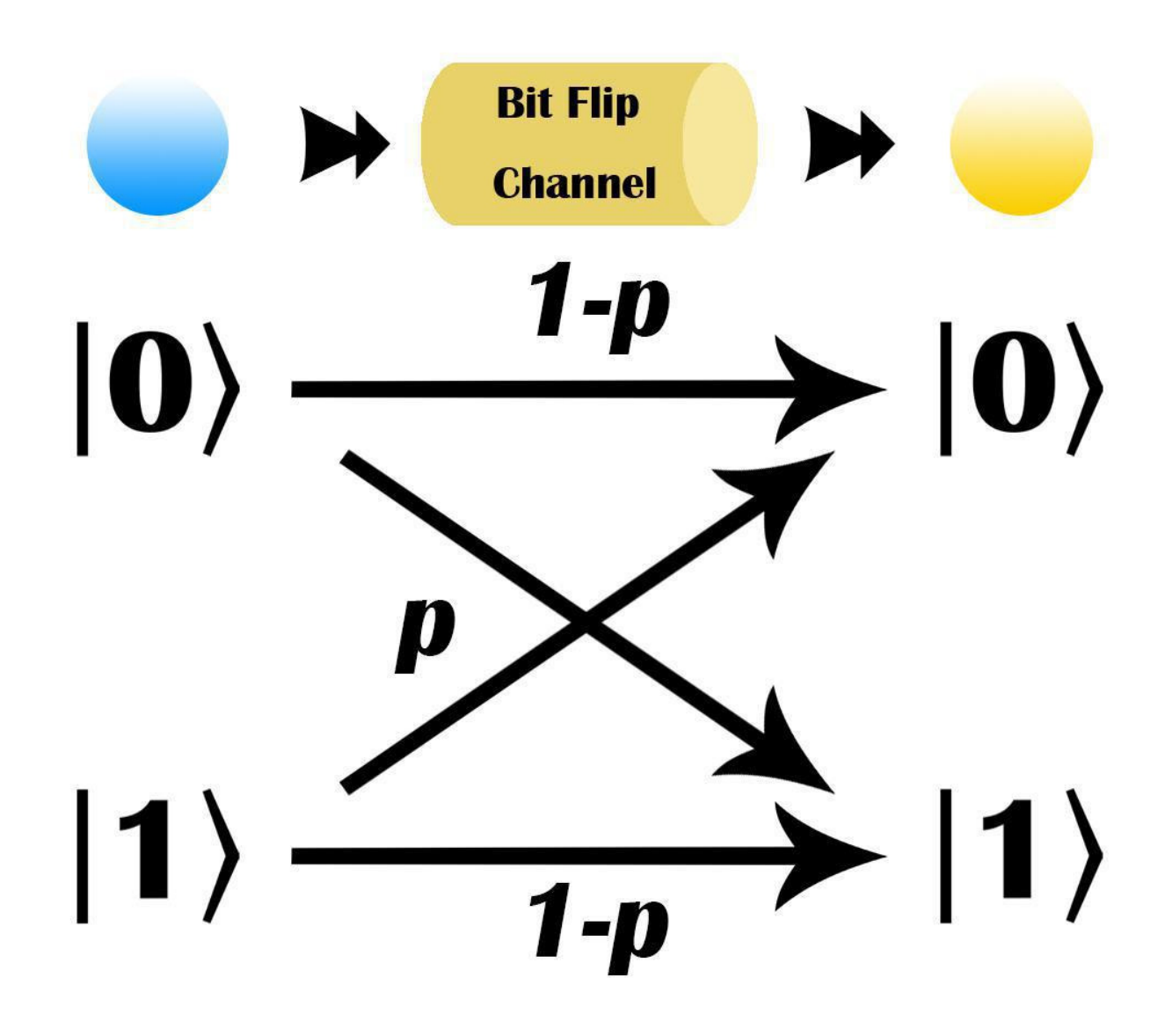}
\caption{Bit flip channel}
\label{fig:Bit_flip_channel}
\end{figure}

\begin{figure}[!t]
\centering
\includegraphics[width=0.5\textwidth]{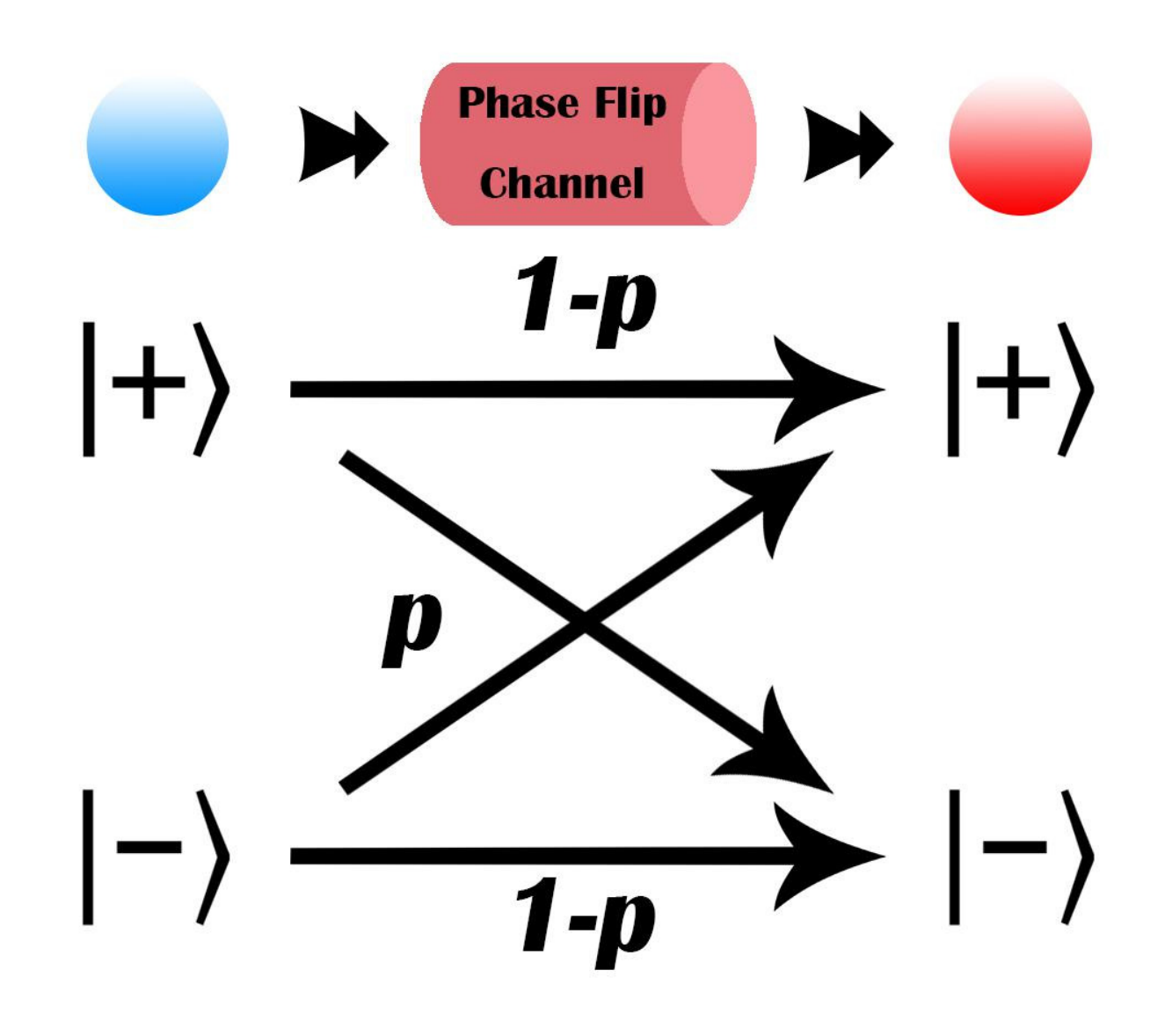}
\caption{Phase flip channel}
\label{fig:Phase_flip_channel}
\end{figure}

Bit flip channel flips $\ket{0}$ and $\ket{1}$ with the same probability $p$, and phase flip channel flips $\ket{+}$ and $\ket{-}$ with the same probability $p$. It’s easy to write out the operation elements\cite{nielsen2002quantum} for them. The operation elements  for bit flip channel is
\begin{equation}
E_{0}=\sqrt{p} I, E_{1}=\sqrt{1-p} X
\end{equation}
where $X$ is the pauli $X$ operator. And the operation elements for phase flip channel is
\begin{equation}
\widetilde{E_{0}}=\sqrt{p} I, \widetilde{E_{1}}=\sqrt{1-p} Z
\end{equation}
where $Z$ is the pauli $Z$ operator.
\subsection{Symmetric coherent information and the MSLCI of two-dimensional-input QQSC}
\label{2.4}
The symmetric coherent information was first proposed in \cite{wilde2013polar}, which is similar to the definition of symmetric capacity used by Arikan\cite{5075875}.

\begin{definition}[\textbf{Symmetric coherent information}]
\label{Uniform coherent information}
For a quantum channel $\mathcal{E}$, the number of whose input qubits is $n$, its input state can be represented by $\rho=\sum_{i=1}^{2^{n}} q_{i}\ket{i}\bra{i}$, its symmetric coherent information $I_{U}$ is defined as the coherent information $I(\rho,\mathcal{E})$ when $q_{1}=q_{2}=\cdots=q_{2^{n}}=1/2^{n}$, namely,
\begin{equation}
I_{U} \equiv I\left(\rho=\sum_{i=1}^{2^{n}} \frac{1}{2^{n}}\ket{i}\bra{i}, \mathcal{E}\right)
\end{equation}
\end{definition}

Arikan has proved that for a classical symmetric channel (actually, the classical symmetric channel mentioned by Arikan in the Part A of Sec. I of \cite{5075875} means a classical binary quasi symmetric channel), the symmetric capacity is its Shannon capacity. However, up to the present, none of the previous studies\cite{6302198,guo2013polar,renes2012efficient,wilde2013polar,hirche2015polar,7208851,7370934,8989387,8815775,6284203,goswami2021quantum} has proved that the symmetric coherent information of a pure quantum channel is its MSLCI. Next, we will prove this theorem for two-dimensional-input QQSC.
\begin{theorem}[\textbf{The MSLCI of two-dimensional-input QQSC}]
\label{theorem2-4:The maximum coherent information of two-dimensional-input QQSC}
The MSLCI of two-dimensional-input QQSC is its symmetric coherent information.
\end{theorem}

\begin{proof}
Assume the input state of a two-dimensional-input QQSC $\mathcal{E}$ is $\rho=q\ket{0}\bra{0}+(1-q)\ket{1}\bra{1}$. According to Theorem \ref{theorem2-3:Operation elements of two-dimensional-input QQSC}, there is a set of operation elements $\left\{E_{k}\right\}$, $1 \leq k \leq M$. By Eq. (\ref{eq:coherent information}) and Eq. (\ref{eq:Se=S(W)}), the coherent information of $\mathcal{E}$ is
\begin{equation}
I(\rho, \mathcal{E})=S(\mathcal{E}(\rho))-S_{e}=S(\mathcal{E}(\rho))-S(W)
\end{equation}
Using Eq. (\ref{Ekij}), one can easily obtain
\begin{equation}
\begin{aligned}
&W_{i j}=tr\left(E_{i} \rho E_{j}^{\dagger}\right)=tr\left(E_{i}(q\ket{0}\bra{0}+(1-q)\ket{1}\bra{1}) E_{j}^{\dagger}\right)\\
&=q\times tr\left(\sqrt{p_{i} p_{j}}\ket{i^{\prime}}\bra{j^{\prime}}\right)+(1-q)\times tr\left(\sqrt{p_{i} p_{j}}\ket{\pi(i)^{\prime}}\bra{\pi(j)^{\prime}}\right)\\
&=p_{i} \delta_{i j}
\end{aligned}
\end{equation}
where $\delta$ is the Kronecker Delta and $\pi$ is a certain permutation.

Hence, $S(W)=H\left(p_{i}\right)$, where $H\left(p_{i}\right)$ is the Shannon entropy of the probability distribution $\left\{p_{1}, \ldots, p_{M}\right\}$. It’s obviously that $S(W)$ has nothing to do with $q$.

Next, we analyze the first term $S(\mathcal{E}(\rho))$. Using Eq. (\ref{Ek}), we get 
\begin{equation}
\begin{aligned}
\label{Maximum of S(E(rho))}
S(\mathcal{E}(\rho))&=S\left(\sum_{k=1}^{M} E_{k} \rho E_{k}^{\dagger}\right)\\
&=S\left(\sum_{k=1}^{M} q p_{k}\ket{k^{'}}\bra{k^{'}}+\sum_{k=1}^{M}(1-q) p_{k}\ket{\pi(k)^{'}}\bra{\pi(k)^{'}}\right)\\
&=S\left(\sum_{k=1}^{M} q p_{k}\ket{k^{'}}\bra{k^{'}}+\sum_{m=1}^{M}(1-q) p_{\pi(m)}\ket{m^{'}}\bra{m^{'}}\right)\\
&=S\left(\sum_{k=1}^{M} q p_{k}\ket{k^{'}}\bra{k^{'}}+\sum_{k=1}^{M}(1-q) p_{\pi(k)}\ket{k^{'}}\bra{k^{'}}\right)
\end{aligned}
\end{equation}

The third equality in Eq.  (\ref{Maximum of S(E(rho))} )holds because $\pi(\pi(k))=k$. If one let $\pi(k)=m$, then $\pi(m)=k$. The last equality is obtained simply by renaming $m$.

Notice that von Neumann entropy has a property  which states that when $\rho_{i}$ have support on orthogonal subspaces, the following equation holds.
\begin{equation}
\label{property}
S\left(\sum_{i} p_{i} \rho_{i}\right)=\sum_{i} p_{i} S\left(\rho_{i}\right)+H\left(p_{i}\right)
\end{equation}
Using Eq.  (\ref{property}), we can further simplify Eq.  (\ref{Maximum of S(E(rho))}).

\begin{equation}
\begin{aligned}
S(\mathcal{E}(\rho))&=\sum_{k=1}^{M}\left[q p_{k}+(1-q) p_{\pi(k)}\right] S\left(\ket{k^{\prime}}\bra{k^{\prime}}\right)+H\left(q p_{k}+(1-q) p_{\pi(k)}\right)\\
&=H\left(q p_{k}+(1-q) p_{\pi(k)}\right)
\end{aligned}
\end{equation}
Taking the derivative with respect to $q$, we obtain
\begin{equation}
\begin{aligned}
\frac{d S(\mathcal{E}(\rho))}{d q}&=-\sum_{k=1}^{M}\left \{ \left(p_{k}-p_{\pi(k)}\right) \log _{2}\left[\left(p_{k}-p_{\pi(k)}\right) q+p_{\pi(k)}\right]+\frac{\left(p_{k}-p_{\pi(k)}\right)}{ \ln 2}\right \} \\
&=-\sum_{k=1}^{M} t_{k}
\end{aligned}
\end{equation}
where $t_{k}=-\left(p_{k}-p_{\pi(k)}\right) \log _{2}\left[\left(p_{k}-p_{\pi(k)}\right) q+p_{\pi(k)}\right]
+\left(p_{k}-p_{\pi(k)}\right)/\ln 2$. Notice that there are $M$ terms in the summation sign, which can be divided into $M/2$ pairs, each of which can be represented by
\begin{equation}
y_{k}=t_{k}+t_{\pi(k)}
\end{equation}
It’s easy to prove that for all $y_{k}$, when $q \in\left[0,\frac{1}{2}\right)$, $y_{k}<0$, when $q\in\left(\frac{1}{2}, 1\right]$, $y_{k}>0$, and when $q=\frac{1}{2}$, $y_{k}=0$. Hence, when $q \in\left[0,\frac{1}{2}\right), \frac{d S(\mathcal{E}(\rho))}{d q}>0$, when $q \in\left(\frac{1}{2},1\right], \frac{d S(\mathcal{E}(\rho))}{d q}<0$, and $q=\frac{1}{2}$, $\frac{d S(\mathcal{E}(\rho))}{d q}=0$. Therefore, $q=\frac{1}{2}$ is the maximum point of $S(\mathcal{E}(\rho))$, which completes the proof.
\end{proof}

\subsection{Quantum channel combining and splitting}
\label{2.5}
\subsubsection{Quantum channel combining}
\label{2.5.1}
Quantum channel combing and splitting are two steps to polarize quantum channels. The quantum channel combing is similar to classical channel combing. Assume the primal channel is $\mathcal{E}: \rho^{Q} \rightarrow \rho^{V}$, which maps the state of a qubit to another state. We denote the input by $Q$, and the output by $V$. As shown in Fig. \ref{fig:Channel_combining_and_splitting}, we use the same recursive manner as in classical channel combining  to combine $N$ primal quantum channels. The difference is that we replace XOR gates in classical channel combining by quantum CNOT gates, and we use quantum SWAP gates to realize the reverse shuffle operator\cite{5075875}. The channel combing process produces the channel $\mathcal{E}_{N}: \rho^{Q_{1} \cdots Q_{N}} \rightarrow \rho^{V_{1}\cdots V_{N}}$, where the subscript $i\ (1\leq i\leq N)$ means the $i$th qubit. This paper follows the Arikan’s rule to denote a row vector, namely, we use $a_{1}^{i}$ as a shorthand for denoting $\left(a_{1}, \cdots, a_{i}\right)$, and the notation $0_{1}^{N}$ is used to denote the all-zero vector. According to this rule, $\mathcal{E}_{N}$ can be rewritten as $\mathcal{E}_{N}: \rho^{Q_{1}^{N}} \rightarrow \rho^{V_{1}^{N}}$.

\begin{figure*}[!t]
\centering
\includegraphics[width=1\textwidth]{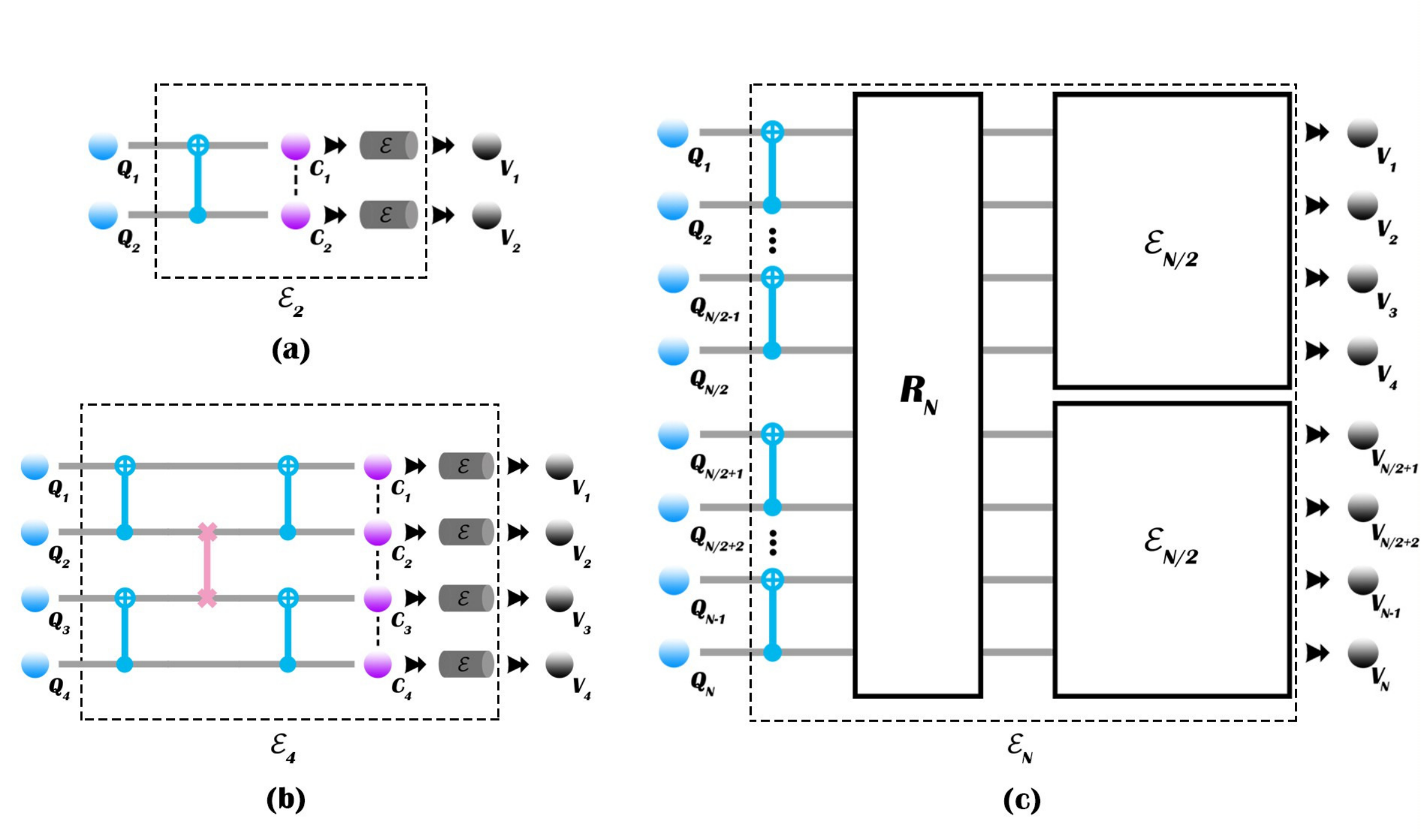}
\caption{(a) Two primal channel $\mathcal{E}$ combines to form channel $\mathcal{E}_{2}$. (b) Two primal $\mathcal{E}_{2}$ combines to form channel $\mathcal{E}_{4}$. (c)Two primal $\mathcal{E}_{N/2}$ combines to form channel $\mathcal{E}_{N}$, $R_{N}$ is the reverse shuffle operator\cite{5075875}. The blue gates are quantum CNOT gates and the pink gates are quantum SWAP gates.}
\label{fig:Channel_combining_and_splitting}
\end{figure*}

\subsubsection{Quantum channel splitting}
\label{2.5.2}
Having combing $N$ quantum channels $\mathcal{E}$ to $\mathcal{E}_{N}$, the next step to polarize quantum channels is splitting $\mathcal{E}_{N}$ to $N$ quantum coordinate channels $\mathcal{E}_{N}^{(i)}: \rho^{Q_{i}} \rightarrow \rho^{V_{1}^{N}, R_{1}^{i-1}}$, namely, $\mathcal{E}_{N}^{(i)}: \rho^{Q_{i}} \rightarrow \rho^{V_{1} \cdots V_{N}, R_{1} \cdots R_{i-1}}$, where $R_{i}$ is the reference system of $Q_{i}$, and $1\leq i\leq N$. The quantum coordinate channels we define is a little bit different from the classical coordinate channels. If we follow the classical definition, the quantum coordinate channels should be $\mathcal{E}_{N}^{(i)}: \rho^{Q_{i}} \rightarrow \rho^{V_{1}^{N}, Q_{1}^{i-1}}$. However, due to quantum no-cloning theorem, $\rho^{Q_{1}^{i-1}}$ and $\rho^{V_{1}^{i-1}}$ cannot appear at the same side. Moreover, according to the manner of Eq. (\ref{eq:pure state QR}) we introduce the reference systems, $R_{i}$ and $Q_{i}$ are in maximally entangled states, which means that the state of $R_{i}$ is the same as $Q_{i}$. Hence, $\mathcal{E}_{N}^{(i)}: \rho^{Q_{i}} \rightarrow \rho^{V_{1}^{N}, R_{1}^{i-1}}$ is a more reasonable definition.

According to Theorem \ref{theorem2-4:The maximum coherent information of two-dimensional-input QQSC}, for a two-dimensional-input QSC, when the input state is a completely mixed state, its SLCI takes maximum.

\section {Symmetry of the quantum combined channel and coordinate channels}
\label{3}
In Sect. \ref{2.5}, quantum combined channel $\mathcal{E}_{N}$ and coordinate channels $\{ \mathcal{E}_{N}^{(i)}\}$ have been defined. The main goal of this section is to prove that if the primal channel $\mathcal{E}$ is a two-dimensional-input QSC with two-dimensional output, then $\mathcal{E}_{N}$ is a QSC and $\{ \mathcal{E}_{N}^{(i)}\}$ are two-dimensional-input QQSCs. We will refer to the proof method which Arikan used  to prove that if the primal binary-input discrete memoryless channel $W: \mathcal{X} \rightarrow \mathcal{Y}$ is symmetric with input alphabet $\mathcal{X}=\{0,1\}$ and output alphabet $\mathcal{Y}=\left\{0^{\prime}, 1^{\prime}\right\}$, classical combined channel $W_{N}: \mathcal{X}^{N} \rightarrow \mathcal{Y}^{N}$ and classical coordinate channels $W_{N}^{(i)}: \mathcal{X} \rightarrow \mathcal{Y}^{N} \times \mathcal{X}^{i-1}$, $1\leq i\leq N$, are symmetric.

\subsection{Symmetry of the quantum combined channel \textbf{$\mathcal{E}_{N}$}}
\label{3.1}
If $\mathcal{E}$ is a two-dimensional-input QSC with two-dimensional output, the BTPM of the channel $\mathcal{E}$ can be expressed as
\begin{equation}
\label{eq:BTPM of Channel E}
\bordermatrix{
		& \ket{0^{'}}                           & \ket{1^{'}}                      \cr
\ket{0} & Pr\left(\ket{0^{'}}\lvert \ket{0}\right)    & Pr\left(\ket{1^{'}}\lvert \ket{0}\right) \cr
\ket{1} & Pr\left(\ket{0^{'}}\lvert \ket{1}\right)    & Pr\left(\ket{1^{'}}\lvert \ket{1}\right)
}
\end{equation}
where $Pr\left(\ket{0^{'}}\lvert \ket{0}\right)=Pr\left(\ket{1^{'}}\lvert \ket{1}\right)$ and $Pr\left(\ket{1^{'}}\lvert \ket{0}\right)=Pr\left(\ket{0^{'}}\lvert \ket{1}\right)$. According to Theorem \ref{theorem2-2:Derive the channel operation elements by BTPM}, we can derive a set of quantum operations $\left\{E_{0}, E_{1}\right\}$ of this channel $\mathcal{E}$, which satisfy $E_{0}\ket{0}=\sqrt{Pr\left(\ket{0^{'}}\lvert \ket{0}\right)}\ket{0^{'}}$, $E_{1}\ket{0}=\sqrt{Pr\left(\ket{1^{'}}\lvert \ket{0}\right)}\ket{0^{'}}$, $E_{0}\ket{1}=\sqrt{Pr\left(\ket{1^{'}}\lvert \ket{1}\right)}\ket{1^{'}}$ and $E_{1}\ket{1}=\sqrt{Pr\left(\ket{0^{'}}\lvert \ket{1}\right)}\ket{0^{'}}$.

\begin{definition}[\textbf{$N$-copy channel $\mathcal{E}^{\otimes N}$ of the primal QSC $\mathcal{E}$}]
\label{N-copy channel}
We define a $N$-copy channel $\mathcal{E}^{\otimes N}: \rho^{Q_{1}} \otimes \ldots \otimes \rho^{Q_{N}} \rightarrow \rho^{V_{1}} \otimes \ldots \otimes \rho^{V_{N}}$ which is simply composed by $N$ independent copies of the primal $\mathcal{E}: \rho^{Q} \rightarrow \rho^{V}$. The operation elements $\left\{F_{k}\right\}$ of $\mathcal{E}^{\otimes N}$ is

\begin{equation}
\label{Fk}
F_{k}=E_{b_{1}}^{1} \otimes E_{b_{2}}^{2} \otimes \cdots \otimes E_{b_{N}}^{N}
\end{equation}
where the subscript $b_{j} \in\{0,1\}$, $1 \leq j \leq N$. The superscript $i$ of $E_{b_{j}}^{i}$ means the operation element $E_{b_{j}}^{i}$ only acts on the $i$th input state $\rho^{Q_{i}}$, and the subscript $k\ (0 \leq k \leq 2^{N}-1)$ of operation elements $F_{k}$ is the decimal number of the binary sequence $b_{1} b_{2} \cdots b_{N}$.
\end{definition}

Assume that $N$ uncorrelated pure input states of the channel $\mathcal{E}^{\otimes N}$ is $\ket{Q_{1}^{N}}=\ket{Q_{1}}\otimes \cdots \otimes\ket{Q_{N}}$, we have
\begin{equation}
\begin{aligned}
F_{k}\ket{Q_{1}^{N}} &=E_{b 1}^{1} \otimes \cdots \otimes E_{b N}^{N}\left(\ket{Q_{1}} \otimes \cdots \otimes \ket{Q_{N}}\right)\\
&=\prod_{i=1}^{N}\sqrt{Pr\left(\ket{V_{i}}\lvert \ket{Q_{i}}\right)}\left(\ket{V_{1}}\otimes\cdots\otimes\ket{V_{N}}\right)\\
&=\sqrt{Pr_{N}\left(\ket{V_{1}^{N}} \lvert \ket{ Q_{1}^{N}}\right)}\ket{V_{1}^{N}}
\end{aligned}
\end{equation}
where we let 
\begin{equation}
\ket{V_{1}}\otimes\cdots\otimes\ket{V_{N}}=\ket{V_{1}^{N}}
\end{equation}
and
\begin{equation}
\label{BTPM of N-copy}
Pr_{N}\left(\ket{V_{1}^{N}}\lvert \ket{Q_{1}^{N}}\right)=\prod_{i=1}^{N} Pr\left( \ket{V_{i}} \lvert \ket{Q_{i}}\right)
\end{equation}
for all $V_{1}^{N} \in \mathcal{Y}^{N}$, $Q_{1}^{N} \in \mathcal{X}^{N}$. $\mathcal{X}^{N}$ is the $N$-power extension alphabet of $\mathcal{X}$ and $\mathcal{Y}^{N}$ is the $N$-power extension alphabet of $\mathcal{Y}$. Eq. (\ref{BTPM of N-copy}) means $Pr_{N}\left(\ket{V_{1}^{N}}\lvert \ket{Q_{1}^{N}}\right)$ is the transition probability when the input state of $\mathcal{E}^{\otimes N}$ is $\ket{Q_{1}^{N}}$ and the output state of $\mathcal{E}^{\otimes N}$ is $\ket{V_{1}^{N}}$. 

In Fig. \ref{fig:Channel_combining_and_splitting}, one can see that $\mathcal{E}^{\otimes N}$ is just the last layer of $\mathcal{E}_{N}$, which is to say, if the recursive combining circuits are omitted, $\mathcal{E}_{N}$ will become $\mathcal{E}^{\otimes N}$. Intuitively, it seems that the BTPM of $\mathcal{E}_{N}$ should have some connections with that of $\mathcal{E}^{\otimes N}$. Next, we prove this intuition. 

\begin{proposition}[\textbf{The BTPM of quantum combined channel $\mathcal{E}_{N}$}]
\label{The BTPM of quantum combined channel}
If each input state of the channel $\mathcal{E}_{N}$ is uncorrelated, the basis transition probabilities of the channel $\mathcal{E}_{N}$ can be obtained by the following equation
\begin{equation}
Pr_{N}\left(\ket{V_{1}^{N}}\lvert \ket{Q_{1}^{N}}\right)=\prod_{i=1}^{N} Pr\left( \ket{V_{i}} \lvert \ket{C_{i}}\right)
\end{equation}
for all $C_{i} \in \mathcal{X}$, $V_{i} \in \mathcal{Y}$, $V_{1}^{N} \in \mathcal{Y}^{N}$, $Q_{1}^{N} \in \mathcal{X}^{N}$, where $\ket{Q_{1}^{N}}$ and $\ket{V_{1}^{N}}$ are the input basis vector and the output basis vector of channel $\mathcal{E}_{N}$ respectively. $\ket{C_{i}}$ and $\ket{V_{i}}$ are the $i$th input basis vector and the $i$th output basis vector of the channel $\mathcal{E}^{\otimes N}$ respectively, as shown in Fig. \ref{fig:Channel_combining_and_splitting}.
\end{proposition}

\begin{proof}
Assume that each uncorrelated input state $\rho^{Q_{i}}$ of the quantum combined channel $\mathcal{E}_{N}$ is $\rho^{Q_{i}}=q\ket{0}\bra{0}+(1-q)\ket{1}\bra{1}$. Then we have
\begin{equation}
\begin{aligned}
\rho^{Q_{1}^{N}}&=\rho^{Q_{1}} \otimes \cdots \otimes \rho^{Q_{N}} \\
&=\left(q\ket{0}\bra{0}+(1-q)\ket{1}\bra{1}\right)^{\otimes N}\\
&=\sum_{Q_{1}^{N} \in \mathcal{X}^{N}} Pr\left(\ket{Q_{1}^{N}}\bra{Q_{1}^{N}}\right)\ket{Q_{1}^{N}}\bra{Q_{1}^{N}}
\end{aligned}
\end{equation}
where alphabet $\mathcal{X}=\{0,1\}$ and $\mathcal{X}^{N}$ is the $N$-power extension alphabet of $\mathcal{X}$, $Pr\left(\ket{Q_{1}^{N}}\bra{Q_{1}^{N}}\right)$  is the probability of $\ket{Q_{1}^{N}}\bra{Q_{1}^{N}}$. Since each input state $\rho^{Q_{i}}$ is uncorrelated with other input states, we have
\begin{equation}
Pr\left(\ket{Q_{1}^{N}}\bra{Q_{1}^{N}}\right)=\prod_{i=1}^{N}Pr\left(\ket{Q_{i}}\bra{Q_{i}}\right)
\end{equation}

Since the process $\ket{Q_{1}^{N}} \rightarrow \ket{C_{1}^{N}}$ which can be seen as an encoding process only includes quantum CNOT gates and quantum SWAP gates, this process must be unitary. As shown in Fig. \ref{fig:Channel_combining_and_splitting}, we use unitary operator $U_{N}$ to denote this encoding process, and obtain
\begin{equation}
\begin{aligned}
\rho^{C_{1}^{N}}&=U_{N} \rho^{Q_{1}^{N}} U_{N}^{\dagger}\\
&=U_{N}\left(\sum_{Q_{1}^{N} \in \mathcal{X}^{N}} Pr\left(\ket{Q_{1}^{N}}\bra{Q_{1}^{N}}\right)\ket{Q_{1}^{N}}\bra{Q_{1}^{N}}\right) U_{N}^{\dagger} \\
&=\sum_{Q_{1}^{N} \in \mathcal{X}^{N}} Pr\left(\ket{Q_{1}^{N}}\bra{Q_{1}^{N}}\right)\ket{Q_{1}^{N}G_{N}}\bra{Q_{1}^{N}G_{N}} \\
&=\sum_{C_{1}^{N} \in \mathcal{X}^{N}} Pr\left(\ket{Q_{1}^{N}}\bra{Q_{1}^{N}}\right)\ket{C_{1}^{N}}\bra{C_{1}^{N}}
\end{aligned}
\end{equation}
where $C_{1}^{N}=Q_{1}^{N}G_{N}$, and $G_{N}$ is generator matrix\cite{5075875}. 

CNOT gates can produce entanglement between its two input qubits while the control qubit is in a superposition state in the computational basis. However, since the input states $\ket{Q_{i}}$ will only take value from $\ket{0}$ or $\ket{1}$, CNOT gates will not produce entanglement\cite{PhysRevLett.113.140401,PhysRevLett.115.020403}. Besides, SWAP gates will not produce entanglement between its two inputs. Hence, all $\ket{C_{i}}$ are uncorrelated. By Eq. (\ref{BTPM of N-copy}), we have
\begin{equation}
\label{BTPM of combined channel}
Pr_{N}\left(\ket{V_{1}^{N}}\lvert \ket{C_{1}^{N}}\right)=\prod_{i=1}^{N} Pr\left( \ket{V_{i}} \lvert \ket{C_{i}}\right)
\end{equation}
Since $\ket{C_{1}^{N}}=\ket{Q_{1}^{N}G_{N}}$, once $Q_{1}^{N}$ is determined, $C_{1}^{N}$ will be determined. Thus,
\begin{equation}
\begin{aligned}
\label{BTPM}
Pr_{N}\left(\ket{V_{1}^{N}}\lvert \ket{Q_{1}^{N}}\right)&=Pr_{N}\left(\ket{V_{1}^{N}}\lvert \ket{Q_{1}^{N}G_{N}}\right)\\
&=Pr_{N}\left(\ket{V_{1}^{N}}\lvert \ket{C_{1}^{N}}\right)\\
&=\prod_{i=1}^{N} Pr\left( \ket{V_{i}} \lvert \ket{C_{i}}\right)
\end{aligned}
\end{equation}
which completes the proof.
\end{proof}

Let $\mathcal{E}: \rho^{Q} \rightarrow \rho^{V}$ is a two-dimensional-input QSC with two-dimensional output. By definition, there is a permutation $\pi_{1}$ on $\mathcal{Y}$ such that 1) $\pi_{1}^{-1}=\pi_{1}$ and 2) $Pr(\ket{V}\lvert \ket{1})=Pr(\ket{\pi_{1}(V)}\lvert \ket{0})$ for all $V \in \mathcal{Y}=\left\{0^{\prime}, 1^{\prime}\right\}$. Let $\pi_{0}$ be the identity permutation on $\mathcal{Y}$. Using the compact notation mentioned by Arikan, we denote $\pi_{Q}(V)$ by $Q \cdot V$, for all $Q \in \mathcal{X}=\{0,1\}$ and $V \in \mathcal{Y}=\left\{0^{\prime}, 1^{\prime}\right\}$.

Observe that $Pr(\ket{V}\lvert \ket{Q \oplus a})=Pr(\ket{a \cdot V}\lvert \ket{Q})$ for all $a, Q \in \mathcal{X}=\{0,1\}$ and $V \in \mathcal{Y}=\left\{0^{\prime}, 1^{\prime}\right\}$. It's easy to verify that $Pr(\ket{V}\lvert \ket{Q \oplus a})=Pr(\ket{(Q \oplus a)\cdot V}\lvert \ket{0})=Pr(\ket{Q \cdot (a\cdot V)}\lvert \ket{0})$ and $Pr(\ket{V}\lvert \ket{Q \oplus a})= Pr(\ket{Q\cdot V}\lvert \ket{a})$ since $\oplus$ is commutative operation on $\mathcal{X}$.

For $Q_{1}^{N} \in \mathcal{X}^{N}$, $V_{1}^{N} \in \mathcal{Y}^{N}$, let
\begin{equation}
Q_{1}^{N} \cdot V_{1}^{N} \triangleq \left(Q_{1} \cdot V_{1}, \cdots, Q_{N} \cdot V_{N}\right)
\end{equation}

Next, we will prove the quantum combined channel $\mathcal{E}_{N}$  is symmetric.

\begin{theorem}[\textbf{the quantum combined channel $\mathcal{E}_{N}$ is a QSC}]
\label{theorem3-1:the quantum combined channel is a QSC}
If the primal channel $\mathcal{E}$ is a two-dimensional-input QSC with two-dimensional output, then the quantum combined channel $\mathcal{E}_{N}$ is QSC in the sense that
\begin{equation}
\begin{aligned}
\label{proof of quantum combined channel}
Pr_{N}\left(\ket{V_{1}^{N}}\lvert \ket{Q_{1}^{N}}\right)=Pr_{N}\left(\ket{a_{1}^{N} G_{N} \cdot V_{1}^{N}}\lvert \ket{Q_{1}^{N} \oplus a_{1}^{N}}\right)
\end{aligned}
\end{equation}
for all $Q_{1}^{N}, a_{1}^{N} \in \mathcal{X}^{N}$ and $V_{1}^{N} \in \mathcal{Y}^{N}$.
\end{theorem}

The Eq. (\ref{proof of quantum combined channel}) means arbitrary row of the BTPM of $\mathcal{E}_{N}$ is a permutation of the first row, and arbitrary column of the BTPM of $\mathcal{E}_{N}$ is a permutation of the first column.

\begin{proof}
By Proposition \ref{The BTPM of quantum combined channel}, we have
\begin{equation}
\begin{aligned}
\label{111}
Pr_{N}\left(\ket{V_{1}^{N}}\lvert \ket{Q_{1}^{N}}\right)&=\prod_{i=1}^{N} Pr\left( \ket{V_{i}} \lvert \ket{C_{i}}\right)\\
&=\prod_{i=1}^{N} Pr\left( \ket{C_{i}\cdot V_{i}} \lvert \ket{0}\right)\\
&=Pr_{N}\left(\ket{C_{1}^{N}\cdot V_{1}^{N}}\lvert \ket{0_{1}^{N}}\right)
\end{aligned}
\end{equation}
Let $b_{1}^{N}=a_{1}^{N} G_{N}$, we have

\begin{equation}
\begin{aligned}
\label{222}
Pr_{N}\left(\ket{b_{1}^{N}\cdot V_{1}^{N}}\lvert \ket{Q_{1}^{N}\oplus a_{1}^{N}}\right)&=Pr_{N}\left(\ket{(C_{1}^{N}\oplus b_{1}^{N})\cdot(b_{1}^{N}\cdot V_{1}^{N})}\lvert \ket{0_{1}^{N}}\right)\\
&=Pr_{N}\left(\ket{C_{1}^{N}\cdot V_{1}^{N}}\lvert \ket{0_{1}^{N}}\right)\\
&=Pr_{N}\left(\ket{V_{1}^{N}}\lvert \ket{Q_{1}^{N}}\right)
\end{aligned}
\end{equation}
which completes the proof.
\end{proof}

\subsection{Symmetry of the quantum coordinate channels \textbf{$\{\mathcal{E}_{N}^{(i)}:0\leq i\leq N\}$}}
\label{3.2}
In this part, we will prove that if the primal channel $\mathcal{E}$ is a two-dimensional-input QSC with two-dimensional output, the coordinate channels $\{\mathcal{E}_{N}^{(i)}:0\leq i\leq N\}$ are QQSCs. The key of the proof is to find out the BTPMs of $\{\mathcal{E}_{N}^{(i)}:0\leq i\leq N\}$, and prove their arbitrary row is a permutation of another row.

\begin{theorem}[\textbf{the quantum coordinate channels $\{\mathcal{E}_{N}^{(i)}:0\leq i\leq N\}$ are QQSCs}]
\label{the quantum coordinate channel is QQSC}
If the primal channel $\mathcal{E}$ is a two-dimensional-input QSC with two-dimensional output, and the input state $\rho^{Q_{i}}=q\ket{0}\bra{0}+(1-q)\ket{1}\bra{1}$, then the arbitrary quantum coordinate channel $\mathcal{E}_{N}^{(i)}: \rho^{Q_{i}} \rightarrow \rho^{V_{1}^{N}, R_{1}^{i-1}}$, $1 \leq i \leq N$, is QQSC. The density operator $\rho^{V_{1}^{N},R_{1}^{i-1}}$ of the joint system $V_{1}^{N},R_{1}^{i-1}$ can be written as
\begin{equation}
\begin{aligned}
\label{rho of VIN,R1i-1}
\rho^{V_{1}^{N},R_{1}^{i-1}}= \sum_{m=0}^{2^{N}-1} \left[q Pr_{N}^{(i)}\left(\ket{m^{'}}\lvert \ket{0}\right)\ket{m^{'}}\bra{m^{'}}+(1-q)  Pr_{N}^{(i)}\left(\ket{m^{'}}\lvert \ket{1}\right)\ket{m^{'}}\bra{m^{'}}\right]
\end{aligned}
\end{equation}
where $\ket{m^{'}}=\sum_{\substack{Q_{1}^{i-1}\\=R_{1}^{i-1} \\ \in \mathcal{X}^{i-1}}} \sqrt{Pr\left(\ket{Q_{1}^{i-1}}\bra{Q_{1}^{i-1}}\right)}\ket{\left(Q_{1}^{i-1}, 0,0_{i+1}^{N}\right) G_{N} \cdot V_{1}^{N}, R_{1}^{i-1}}$, $0 \leq m \leq 2^{N}-1$, form a set of basis $\{ \ket{m^{'}}\}_{m=0, \cdots, 2^{N}-1}$ which contains $2^{N}$ basis vectors. And the basis transition probabilities are
\begin{equation}
\label{BTPM of rho of VIN,R1i-1}
\begin{aligned}
&Pr_{N}^{(i)}\left(\ket{m^{'}}\lvert \ket{Q_{i}}\right)\\
&=Pr_{N}^{(i)}\left(\sum_{\substack{Q_{1}^{i-1}\\ =R_{1}^{i-1} \\ \in \mathcal{X}^{i-1}}} \sqrt{Pr\left(\ket{Q_{1}^{i-1}}\bra{Q_{1}^{i-1}}\right)}\ket{\left(Q_{1}^{i-1}, 0,0_{i+1}^{N}\right) G_{N} \cdot V_{1}^{N}, R_{1}^{i-1}}\lvert \ket{Q_{i}}\right)\\
&=\sum_{Q_{i+1}^{N}\in \mathcal{X}^{N-i}}Pr\left(\ket{Q_{i+1}^{N}}\bra{Q_{i+1}^{N}}\right)Pr_{N}\left(\ket{V_{1}^{N}}\lvert \ket{0_{1}^{i-1},Q_{i},Q_{i+1}^{N}}\right)\\
&=Pr_{N}^{(i)}\left(\sum_{Q_{1}^{i-1}=R_{1}^{i-1} \in \mathcal{X}^{i-1}} \sqrt{Pr\left(\ket{Q_{1}^{i-1}}\bra{Q_{1}^{i-1}}\right)}\right.\\
&\left.\ket{\left(a_{1}^{i-1}, 1,a_{i+1}^{N}\oplus Q_{1}^{i-1}, 0,0_{i+1}^{N}\right) G_{N} \cdot V_{1}^{N}, R_{1}^{i-1}}\lvert \ket{Q_{i}\oplus 1}\right)
\end{aligned}
\end{equation}
for all $V_{1}^{N} \in \mathcal{Y}^{N}$, $Q_{i} \in \mathcal{X}$, $(a_{1}^{i-1},1,a_{i+1}^{N})$, $(Q_{1}^{i-1},Q_{i},Q_{i+1}^{N})\in \mathcal{X}^{N}$, $N=2^{n}$, $n \geq 0$, $1 \leq i \leq N$, which means arbitrary row of the BTPM of $\mathcal{E}_{N}^{(i)}$ is a permutation of another row.
\end{theorem}

The proof of Theorem \ref{the quantum coordinate channel is QQSC} is given in Appendix \ref{Proof of Theorem 6}.

Since $\mathcal{E}_{N}^{(i)}$ is two-dimensional-input QQSC, according to Theorem \ref{theorem2-4:The maximum coherent information of two-dimensional-input QQSC}, the MSLCI of $\mathcal{E}_{N}^{(i)}$ is equal to its symmetric coherent information, namely, the SLCI of $\mathcal{E}_{N}^{(i)}$ takes the maximum when the input state is $\rho^{Q_{i}}=\frac{1}{2}\ket{0}\bra{0}+\frac{1}{2}\ket{1}\bra{1}$, therefore Eq. (\ref{rho of VIN,R1i-1}) and Eq. (\ref{BTPM of rho of VIN,R1i-1}) are reduced to
\begin{equation}
\begin{aligned}
\label{new rho of VIN,R1i-1}
\rho^{V_{1}^{N}, R_{1}^{i-1}}&=\frac{1}{2}\sum_{m=0}^{2^{N}-1} Pr_{N}^{(i)}\left(\ket{m}\lvert \ket{0}\right)\ket{m}\bra{m}+\frac{1}{2} \sum_{m=0}^{2^{N}-1} Pr_{N}^{(i)}\left(\ket{m}\lvert \ket{1}\right)\ket{m}\bra{m}
\end{aligned}
\end{equation}
and
\begin{equation}
\label{new BTPM of rho of VIN,R1i-1}
\begin{aligned}
&Pr_{N}^{(i)}\left(\ket{m}\lvert \ket{Q_{i}}\right)\\
&=Pr_{N}^{(i)}\left(\sum_{Q_{1}^{i-1}=R_{1}^{i-1} \in \mathcal{X}^{i-1}} \frac{1}{2^{\frac{i-1}{2}}}\ket{\left(Q_{1}^{i-1}, 0,0_{i+1}^{N}\right) G_{N} \cdot V_{1}^{N}, R_{1}^{i-1}}\lvert \ket{Q_{i}}\right)\\
&= \frac{2^{i-1}}{2^{N-1}}\sum_{Q_{i+1}^{N}\in \mathcal{X}^{N-i}}Pr_{N}\left(\ket{V_{1}^{N}}\lvert \ket{0_{1}^{i-1},0,Q_{i+1}^{N}}\right)\\
&=Pr_{N}^{(i)}\left(\sum_{\substack{Q_{1}^{i-1}\\=R_{1}^{i-1}\\ \in \mathcal{X}^{i-1}}} \frac{1}{2^{\frac{i-1}{2}}}\ket{\left(a_{1}^{i-1}, 1,a_{i+1}^{N}\oplus Q_{1}^{i-1}, 0,0_{i+1}^{N}\right) G_{N} \cdot V_{1}^{N}, R_{1}^{i-1}}\lvert \ket{Q_{i}\oplus 1}\right)
\end{aligned}
\end{equation}
where $\ket{m}=\sum_{Q_{1}^{i-1}=R_{1}^{i-1} \in \mathcal{X}^{i-1}} \frac{1}{2^{\frac{i-1}{2}}}\left\lvert \left(Q_{1}^{i-1}, 0,0_{i+1}^{N}\right) G_{N}\cdot V_{1}^{N}, R_{1}^{i-1}\right\rangle$, $0 \leq m \leq 2^{N}-1$.

\section {Polarization of two-dimensional-input QSC}
\label{4}
The goal of this section is to prove the MSLCI of coordinate channels $\{ \mathcal{E}_{N}^{(i)}\}$ will polarize.

One can see that the quantum combined channel $\mathcal{E}_{N}$ corresponds to a classical combined channel $W_{N}$, which is obtained by simply replace the quantum circuits in Fig. \ref{fig:Channel_combining_and_splitting} to a classical ones, and the primal channel $\mathcal{E}$ to a classical channel $W$. Our proof in this section makes use of the connection between $\mathcal{E}_{N}$ and  $W_{N}$.

If the BTPM of the primal QSC $\mathcal{E}$ and the TPM of classical primal BSC $W$ are the same, first of all, we prove that  the BTPM of the quantum combined channel $\mathcal{E}_{N}$ and the TPM of classical combined channel $W_{N}$ are the same; secondly, we prove the BTPM of quantum coordinate channel $\mathcal{E}_{N}^{(i)}$ can be derived from the TPM of classical coordinate channel $W_{N}^{(i)}$ which reveals the relationship between the BTPM of $\mathcal{E}_{N}^{(i)}$ and the TPM of $W_{N}^{(i)}$; finally we use this relationship to prove that the MSLCI of $\mathcal{E}_{N}^{(i)}$ numerically equals to the Shannon capacity of $W_{N}^{(i)}$. Since the Shannon capacity of $\{W_{N}^{(i)}\}$ will polarize, the MSLCI of $\{ \mathcal{E}_{N}^{(i)}\}$ will polarize as well. Moreover, due to the MSLCI of the primal channel $\mathcal{E}$ being equal to the Shannon capacity of the classical primal channel $W$, the polarization rate of $\{ \mathcal{E}_{N}^{(i)}\}$ equals to the MSLCI of $\mathcal{E}$, which is referred to Arikan's method\cite{5075875}.

\begin{proposition}[\textbf{Relationship between the BTPM of $\mathcal{E}_{N}$ and the TPM of $W_{N}$}]
\label{BTPM and TPM are the same}
Assume that the BTPM of a two-dimensional-input QSC with two-dimensional output $\mathcal{E}$ is
\begin{equation}
\bordermatrix{
		& \ket{0^{'}}                           & \ket{1^{'}}                      \cr
\ket{0} & Pr\left(\ket{0^{'}}\lvert \ket{0}\right)    & Pr\left(\ket{1^{'}}\lvert \ket{0}\right) \cr
\ket{1} & Pr\left(\ket{0^{'}}\lvert \ket{1}\right)    & Pr\left(\ket{1^{'}}\lvert \ket{1}\right)
}
\end{equation}
where $Pr\left(\ket{0^{'}}\lvert \ket{0}\right)=Pr\left(\ket{1^{'}}\lvert \ket{1}\right)=W\left(0^{'}\lvert 0\right)=W\left(1^{'}\lvert 1\right)$ and $Pr\left(\ket{1^{'}}\lvert \ket{0}\right)=Pr\left(\ket{0^{'}}\lvert \ket{1}\right)=W\left(1^{'}\lvert 0\right)=W\left(0^{'}\lvert 1\right)$. Then the BTPM of quantum combined channel $\mathcal{E}_{N}$ and the TPM of classical combined channel $W_{N}$ are the same, that is to say
\begin{equation}
Pr_{N}\left(\ket{V_{1}^{N}}\lvert \ket{Q_{1}^{N}}\right)=W_{N}\left(y_{1}^{N} \lvert u_{1}^{N}\right)
\end{equation}
for all $V_{1}^{N}=y_{1}^{N} \in \mathcal{Y}^{N}$ and $Q_{1}^{N}=u_{1}^{N} \in \mathcal{X}^{N}$, where $y, V \in \mathcal{Y}=\left\{0^{\prime}, 1^{\prime}\right\}$ and $u, Q \in \mathcal{X}=\{0,1\}$.
\end{proposition}

\begin{proof}
By Proposition \ref{The BTPM of quantum combined channel}, we have
\begin{equation}
\begin{aligned}
Pr_{N}\left(\ket{V_{1}^{N}}\lvert \ket{Q_{1}^{N}}\right)&=Pr_{N}\left(\ket{V_{1}^{N}}\lvert \ket{Q_{1}^{N}G_{N}}\right)\\
&=Pr_{N}\left(\ket{V_{1}^{N}}\lvert \ket{C_{1}^{N}}\right)\\
&=\prod_{i=1}^{N} Pr\left( \ket{V_{i}} \lvert \ket{C_{i}}\right)
\end{aligned}
\end{equation}
According to Arikan’s method\cite{5075875}, we have
\begin{equation}
\begin{aligned}
W_{N}\left(y_{1}^{N}\lvert u_{1}^{N}\right)&=W_{N}\left(y_{1}^{N}\lvert u_{1}^{N}G_{N}\right)\\
&=W_{N}\left(y_{1}^{N}\lvert x_{1}^{N}\right)\\
&=\prod_{i=1}^{N} W\left(V_{i} \lvert x_{i}\right)
\end{aligned}
\end{equation}
where $u_{1}^{N} G_{N}=x_{1}^{N}$. Since $V_{1}^{N}=y_{1}^{N}$ and $Q_{1}^{N}=u_{1}^{N}$, then we have $Q_{1}^{N} G_{N}=u_{1}^{N} G_{N}=C_{1}^{N}=x_{1}^{N}$. Thus, we have $Pr\left( \ket{V_{i}} \lvert \ket{C_{i}}\right)=W\left(y_{i} \lvert x_{i}\right)$, and obtain
\begin{equation}
\prod_{i=1}^{N} Pr\left( \ket{V_{i}} \lvert \ket{C_{i}}\right)=\prod_{i=1}^{N} W\left(V_{i} \lvert x_{i}\right)
\end{equation}
which completes the proof.
\end{proof}

\begin{proposition}[\textbf{Relationship between the BTPM of $\mathcal{E}_{N}^{(i)}$ and the TPM of $W_{N}^{(i)}$}]
\label{BTMP can be derived from TPM}
According to Eq. (\ref{new rho of VIN,R1i-1}) and Eq. (\ref{new BTPM of rho of VIN,R1i-1}), when the input state $\rho^{Q_{i}}$ of the channel $\mathcal{E}_{N}^{(i)}$ is $\rho^{Q_{i}}=\frac{1}{2}\ket{0}\bra{0}+\frac{1}{2}\ket{1}\bra{1}$, the output state $\ket{m}$ of the channel $\mathcal{E}_{N}^{(i)}$ is
\begin{equation}
\ket{m}=\sum_{Q_{1}^{i-1}\\=R_{1}^{i-1} \in \mathcal{X}^{i-1}} \frac{1}{2^{\frac{i-1}{2}}}\ket{\left(Q_{1}^{i-1}, 0, 0_{i+1}^{N}\right) G_{N} \cdot V_{1}^{N}, R_{1}^{i-1}}
\end{equation}
and the basis transition probabilities are
\begin{equation}
\begin{aligned}
&Pr_{N}^{(i)}\left(\ket{m}\lvert \ket{Q_{i}}\right)=\\
&Pr_{N}^{(i)}\left(\sum_{Q_{1}^{i-1}=R_{1}^{i-1} \in \mathcal{X}^{i-1}} \frac{1}{2^{\frac{i-1}{2}}}\ket{\left(Q_{1}^{i-1}, 0, 0_{i+1}^{N}\right) G_{N} \cdot V_{1}^{N}, R_{1}^{i-1}}\lvert \ket{Q_{i}}\right)\\
&=\frac{2^{i-1}}{2^{N-1}}\sum_{Q_{1}^{i-1}\in \mathcal{X}^{i-1}}Pr_{N}\left(\ket{V_{1}^{N}}\lvert\ket{0_{1}^{i-1}, Q_{i}, Q_{i+1}^{N}}\right)
\end{aligned}
\end{equation}
We can derive $P r_{N}^{(i)}\left(\ket{m}\lvert \ket{Q_{i}}\right)$ from the TPM of classical  coordinate channels $W_{N}^{(i)}$
\begin{equation}
\begin{aligned}
&Pr_{N}^{(i)}\left(\ket{m}\lvert \ket{Q_{i}}\right)\\
&=\sum_{u_{1}^{i-1}\in \mathcal{X}^{i-1}}\frac{1}{2^{N-1}}\sum_{u_{i+1}^{N}\in \mathcal{X}^{N-i}}W_{N}\left( (u_{1}^{i-1},0,0_{i+1}^{N})G_{N}\cdot y_{1}^{N}\lvert u_{1}^{i-1},u_{i},u_{i+1}^{N}\right)\\
&=\frac{2^{i-1}}{2^{N-1}}\sum_{u_{i+1}^{N}\in \mathcal{X}^{N-i}}W_{N}\left(y_{1}^{N}\lvert 0_{1}^{i-1},u_{i},u_{i+1}^{N}\right)\\
&=2^{i-1}W_{N}^{(i)}\left(y_{1}^{N}, 0_{1}^{i-1}\lvert u_{i}\right)
\end{aligned}
\end{equation}
for all $V_{1}^{N}=y_{1}^{N} \in \mathcal{Y}^{N}$ and $Q_{1}^{N}=u_{1}^{N} \in \mathcal{X}^{N}$, where $y, V \in \mathcal{Y}=\left\{0^{\prime}, 1^{\prime}\right\}$ and $u, Q \in \mathcal{X}=\{0,1\}$.
\end{proposition}

The proof of Propositon \ref{BTMP can be derived from TPM} is given in Appendix \ref{Proof of Proposition 3}.

The Proposition \ref{BTMP can be derived from TPM} means that arbitrary column of the BTPM of each $\mathcal{E}_{N}^{(i)}$ is the sum of some $2^{i-1}$ columns of the TPM whose corresponding elements are all equal, hence the TPM of each $W_{N}^{(i)}$ has $2^{N+i-1}$ columns while the BTPM of each $\mathcal{E}_{N}^{(i)}$ has $2^{N}$ columns.

\begin{theorem}[\textbf{the polarization of quantum coordinate channels $\{\mathcal{E}_{N}^{(i)}\}$}]
\label{the polarization of quantum coordinate channels}
If the BTPM of the primal QSC $\mathcal{E}$ and the TPM of classical primal BSC $W$ are the same, the MSLCI $I\left(\rho^{Q_{i}}, \mathcal{E}_{N}^{(i)}\right)$ of the quantum coordinate channel $\mathcal{E}_{N}^{(i)}$ is numerically equal to the Shannon capacity $I\left(W_{N}^{(i)}\right)$ of classical coordinate channel $W_{N}^{(i)}$, namely,
\begin{equation}
\begin{aligned}
&I\left(\rho^{Q_{i}}, \mathcal{E}_{N}^{(i)}\right)=S\left(\rho^{V_{1}^{N},R_{1}^{i-1}}\right)-S\left(\rho^{V_{1}^{N},R_{1}^{i}}\right)=I\left(W_{N}^{(i)}\right)\\
&=H\left(y_{1}^{N} u_{1}^{i-1}\right)-H\left(y_{1}^{N} u_{1}^{i}\right)+H\left(u_{i}\right)
\end{aligned}
\end{equation}
where $p\left(u_{i}=0\right)=p\left(u_{i}=1\right)=\frac{1}{2}$ and the density operator $\rho^{Q_{i}}=\frac{1}{2}\ket{0}\bra{0}+\frac{1}{2}\ket{1}\bra{1}$ is the input of the quantum coordinate channel $\mathcal{E}_{N}^{(i)}$, which is also the $i$th input of the quantum combined channel $\mathcal{E}_{N}$. $S(\cdot)$ is von Neumann entropy, and $H(\cdot)$ is Shannon entropy. Since classical coordinate channels $\{W_{N}^{(i)}\}$ polarize, quantum coordinate channels $\{ \mathcal{E}_{N}^{(i)}\}$ polarize as well.
\end{theorem}

\begin{proof}
For $W_{N}^{(i)}$, its Shannon capacity is $I\left(W_{N}^{(i)}\right)=H\left(y_{1}^{N} u_{1}^{i-1}\right)-H\left(y_{1}^{N} u_{1}^{i}\right)+H\left(u_{i}\right)$. To calculate the Shannon capacity of $W_{N}^{(i)}$, we should first calculate $H\left(y_{1}^{N} u_{1}^{i-1}\right)$ which is the the Shannon entropy of the output $y_{1}^{N} u_{1}^{i-1}$ of $W_{N}^{(i)}$
\begin{equation}
\begin{aligned}
H\left(y_{1}^{N} u_{1}^{i-1}\right)=\sum_{y_{1}^{N} u_{1}^{i-1} \in \mathcal{Y}^{N} \times \mathcal{X}^{i-1}}-p\left(y_{1}^{N}, u_{1}^{i-1}\right) \log _{2} p\left(y_{1}^{N}, u_{1}^{i-1}\right)
\end{aligned}
\end{equation}
where $p\left(y_{1}^{N}, u_{1}^{i-1}\right)=\frac{1}{2} W_{N}^{(i)}\left(y_{1}^{N}, u_{1}^{i-1} \mid u_{i}=0\right)+\frac{1}{2} W_{N}^{(i)}\left(y_{1}^{N}, u_{1}^{i-1} \mid u_{i}=1\right)$.
Notice that  by Proposition \ref{BTMP can be derived from TPM}, we have $W_{N}^{(i)}\left(y_{1}^{N}, 0_{1}^{i-1} \mid u_{i}\right)=W_{N}^{(i)}\left(\left(u_{1}^{i-1},0,0_{i+1}^{N}\right) G_{N} \cdot y_{1}^{N}, 0_{1}^{i-1} \oplus u_{1}^{i-1} \mid u_{i}\right)$ for all $u_{1}^{i-1} \in \mathcal{X}^{i-1}$, hence
\begin{equation}
\begin{aligned}
H\left(y_{1}^{N} u_{1}^{i-1}\right)=2^{i-1} \sum_{y_{1}^{N} \in \mathcal{Y}^{N}}-p\left(y_{1}^{N}, 0_{1}^{i-1}\right) \log _{2} p\left(y_{1}^{N}, 0_{1}^{i-1}\right)
\end{aligned}
\end{equation}
and
\begin{equation}
\begin{aligned}
\sum_{y_{1}^{N} u_{1}^{i-1} \in \mathcal{Y}^{N} \times\mathcal{X}^{i-1}} p\left(y_{1}^{N}, u_{1}^{i-1}\right)=2^{i-1} \sum_{y_{1}^{N} \in \mathcal{Y}^{N}} p\left(y_{1}^{N}, 0_{1}^{i-1}\right)=1
\end{aligned}
\end{equation}
Now, we calculate $S\left(\rho^{V_{1}^{N},R_{1}^{i-1}}\right)$, the von Neumann entropy of the ouput state $\rho^{V_{1}^{N},R_{1}^{i-1}}$ of $ \mathcal{E}_{N}^{(i)}$. By Eq. (\ref{new rho of VIN,R1i-1}), we have
\begin{equation}
S\left(\rho^{V_{1}^{N},R_{1}^{i-1}}\right)=-\sum_{m} p\left( \ket{m} \right) \log _{2} p\left(\ket{m}\right)
\end{equation}
where $p\left( \ket{m} \right)=\frac{1}{2}Pr_{N}^{(i)}\left(\ket{m}\lvert \ket{0}\right)+\frac{1}{2}Pr_{N}^{(i)}\left(\ket{m}\lvert \ket{1}\right)$. By Proposition \ref{BTMP can be derived from TPM}, we have
\begin{equation}
\begin{aligned}
Pr_{N}^{(i)}\left(\ket{m}\lvert \ket{Q_{i}}\right)&=Pr_{N}^{(i)}\left(\sum_{\substack{Q_{1}^{i-1}\\=R_{1}^{i-1}\\ \in \mathcal{X}^{i-1}}} \frac{1}{2^{\frac{i-1}{2}}}\ket{\left(Q_{1}^{i-1}, 0, 0_{i+1}^{N}\right) G_{N} \cdot V_{1}^{N}, R_{1}^{i-1}}\lvert \ket{Q_{i}}\right)\\
&=2^{i-1} W_{N}^{(i)}\left(y_{1}^{N}, 0_{1}^{i-1} \mid u_{i}\right)
\end{aligned}
\end{equation}
for all $y_{1}^{N}=V_{1}^{N} \in \mathcal{Y}^{N}$ and $Q_{1}^{N}=u_{1}^{N} \in \mathcal{X}^{N}$. Thus $S\left(\rho^{V_{1}^{N},R_{1}^{i-1}}\right)$ can be rewritten as
\begin{equation}
\begin{aligned}
&S\left(\rho^{V_{1}^{N},R_{1}^{i-1}}\right)=-\sum_{y_{1}^{N} \in \mathcal{Y}^{N}} 2^{i-1} p\left(y_{1}^{N}, 0_{1}^{i-1}\right) \log _{2}\left[2^{i-1} p\left(y_{1}^{N}, 0_{1}^{i-1}\right)\right] \\
&=-2^{i-1} \sum_{y_{1}^{N} \in \mathcal{Y}^{N}} p\left(y_{1}^{N}, 0_{1}^{i-1}\right) \log _{2} p\left(y_{1}^{N}, 0_{1}^{i-1}\right)-(i-1) \times 2^{i-1} \sum_{y_{1}^{N} \in \mathcal{Y}^{N}} p\left(y_{1}^{N}, 0_{1}^{i-1}\right)\\
&=H\left(y_{1}^{N} u_{1}^{i-1}\right)-(i-1)
\end{aligned}
\end{equation}
Using the same method, we have $S\left(\rho^{V_{1}^{N},R_{1}^{i}}\right)=H\left(y_{1}^{N} u_{1}^{i}\right)-i$. Thus $I\left(\rho^{Q_{i}}, \mathcal{E}_{N}^{(i)}\right)=S\left(\rho^{V_{1}^{N},R_{1}^{i-1}}\right)-S\left(\rho^{V_{1}^{N} ,R_{1}^{i}}\right)=H\left(y_{1}^{N} u_{1}^{i-1}\right)-H\left(y_{1}^{N} u_{1}^{i}\right)+1$. Notice that $H\left(u_{i}\right)=H\left(\frac{1}{2}\right)=1$, thus we have 
\begin{equation}
\begin{aligned}
I\left(\rho^{Q_{i}},\mathcal{E}_{N}^{(i)}\right)&=H\left(y_{1}^{N} u_{1}^{i-1}\right)-H\left(y_{1}^{N} u_{1}^{i}\right)+H\left(u_{i}\right)=I\left(W_{N}^{(i)}\right)
\end{aligned}
\end{equation}
which completes the proof.
\end{proof}

\section {Conclusion}
\label{5}
The core of this paper is to prove that there is a polarization phenomenon in quantum channels similar to classical channel polarization. To prove this, we first define BTPM, and show how to use BTPM to determine a set of operation elements of a quantum channel. Then we use BTPM to define QSC and QQSC, and prove that the MSLCI of two-dimensional-input QQSC is its symmetric coherent information, which was not proved before our work. After this, we introduce the quantum channel combining and splitting, and obtain the quantum combined channel $\mathcal{E}_{N}$ and coordinate channels $\{ \mathcal{E}_{N}^{(i)}\}$. It has been proved in Sect. \ref{3} that if the primal channel $\mathcal{E}$ is a two-dimensional-input QSC, then $\mathcal{E}_{N}$ is a two-dimensional-input QSC and $\{ \mathcal{E}_{N}^{(i)}\}$ are two-dimensional-input QQSCs. Based on the above work, we prove that the MSLCI of the coordinate channels will polarize – some of them tend to 1 while the others tend to 0 with the increase of $N$, and the ratio of the former to $N$ equals to the MSLCI of the primal channel $\mathcal{E}$, which completes the proof that there is a polarization phenomenon in quantum channels. 

However, whether we can make use of this polarization phenomenon of quantum channels to design a quantum error correcting code which can achieve the MSLCI of QSC is still unknow.

\backmatter

%
%
%
%
%
%

\section*{Declarations}

\begin{itemize}
\item Funding

Not applicable.
\item Competing interests

The authors have no competing interests to declare that are relevant to the content of this article.
\item Ethics approval 

Not applicable.
\item Consent to participate

Not applicable.
\item Consent for publication

Not applicable.
\item Availability of data and materials

Data sharing not applicable to this article as no datasets were generated or analysed during the current study.
\item Code availability 

Not applicable.
\item Authors' contributions

All authors conceived the work, analysed the results and wrote the manuscript.
\end{itemize}


\begin{appendices}
\section {A Particular Rule}
\label{A Particular Rule}

Before proving Theorem \ref{the quantum coordinate channel is QQSC}, we make a particular rule which will be used in the second step of the proof.

This rule is used to label the operator elements of a channel through a one-to-one relationship between operator elements and output states. First, we fixed the input state $\ket{Q_{1}^{N}}$ of the quantum combined channel $\mathcal{E}_{N}$ to $\ket{0_{1}^{N}}$, then arbitrary operator element $F_{k} \in\left\{F_{k}\right\}_{k=0, \cdots, 2^{N}-1}$ of the $N$-copy channel $\mathcal{E}^{\otimes N}$ uniquely corresponds to a output state $\ket{V_{1}^{N}}$, $V_{1}^{N}\in \mathcal{Y}^{N}$, namely,
\begin{equation}
\label{regulation 1}
F_{k}\ket{0_{1}^{N}G_{N}}=\sqrt{Pr_{N}\left(\ket{V_{1}^{N}}\lvert \ket{0_{1}^{N}}\right)}\ket{V_{1}^{N}}
\end{equation}
By Definition \ref{N-copy channel}, we have
\begin{equation}
\label{Fk}
F_{k}=E_{b_{1}}^{1} \otimes E_{b_{2}}^{2} \otimes \cdots \otimes E_{b_{N}}^{N}
\end{equation}
the subscript $k$ of $F_{k}$ is the decimal number of the binary sequence $b_{1} b_{2} \cdots b_{N}$.

To further understanding this rule, we take $2$-copy channel $\mathcal{E}^{\otimes 2}$ for example, and primal channel $\mathcal{E}$ is Bit flip channel whose operator elements are $\{ E_{0}=\sqrt{p}X,E_{1}=\sqrt{1-p}I\}$. It's easy to obtain that four operator elements of $\mathcal{E}^{\otimes 2}$ are $F_{0}=pX^{1}\otimes X^{2}$, $F_{1}=\sqrt{p(1-p)}X^{1}\otimes I^{2}$, $F_{2}=\sqrt{p(1-p)}I^{1}\otimes X^{2}$ and $F_{3}=(1-p)I^{1}\otimes I^{2}$, respectively. Assume that the input state of primal channel $\mathcal{E}$ will only take value from $\ket{0}=\left(\begin{array}{l}1 \\ 0\end{array}\right)$ or $\ket{1}=\left(\begin{array}{l}0 \\ 1\end{array}\right)$. Then the input space $\{\ket{Q_{1}^{2}}\}$ of the quantum combined channel $\mathcal{E}_{2}$ must be $\{\ket{Q_{1}^{2}}\}=\{\ket{00},\ket{01},\ket{10},\ket{11}\}$, and the output space $\{\ket{V_{1}^{2}}\}$ of the quantum combined channel $\mathcal{E}_{2}$ must be $\{\ket{V_{1}^{2}}\}=\{\ket{00},\ket{01},\ket{10},\ket{11}\}$, which means different operator element $F_{k},\ 0\leq k \leq 3$, will map the input space $\{\ket{Q_{1}^{2}}\}$ to the same output space $\{\ket{V_{1}^{2}}\}$. Thus we fixed the input state to $\ket{00}$, and a one-to-one relationship between operator element $F_{k}(0\leq k\leq 3)$ and output state $\ket{V_{1}^{2}}$ of the channel $\mathcal{E}^{\otimes 2}$ is established, namely, $F_{0}$ corresponds to $\ket{11}$, $F_{1}$ corresponds to $\ket{10}$, $F_{2}$ corresponds to $\ket{01}$ and $F_{3}$ corresponds to $\ket{00}$.

By using Theorem \ref{theorem3-1:the quantum combined channel is a QSC} and Eq.  (\ref{regulation 1}), we have
\begin{equation}
\begin{aligned}
\label{regulation 2}
F_{k}\ket{Q_{1}^{N}G_{N}}=\sqrt{Pr_{N}\left( \ket{Q_{1}^{N}G_{N} \cdot V_{1}^{N}}\lvert \ket{Q_{1}^{N}}\right)}\ket{Q_{1}^{N}G_{N} \cdot V_{1}^{N}}
\end{aligned}
\end{equation}
for all $Q_{1}^{N}\in \mathcal{X}^{N}$ and $V_{1}^{N}\in \mathcal{Y}^{N}$.

\section {Proof of Theorem 7}
\label{Proof of Theorem 6}
In this section, we prove Theorem \ref{the quantum coordinate channel is QQSC} that the quantum coordinate channels $\{\mathcal{E}_{N}^{(i)}\}$ are QQSCs. At the second step of the proof, we use the particular rule that we make in Appendix \ref{A Particular Rule}.

\begin{proof}
In subsection \ref{2.5}, we define quantum coordinate channel $\mathcal{E}_{N}^{(i)}$, $1\leq i \leq N$, whose input is $\rho^{Q_{i}}$ and output is $\rho^{V_{1}^{N},R_{1}^{i-1}}$.

\textbf{1. The first step of the proof: obtain the general form of density operator $\rho^{V_{1}^{N},R_{1}^{i-1}}$ of quantum joint system $V_{1}^{N},R_{1}^{i-1}$.}

Assume that each input state $\rho^{Q_{i}}$ of the quantum combined channel $\mathcal{E}_{N}$ is $\rho^{Q_{i}}=q\ket{0}\bra{0}+(1-q)\ket{1}\bra{1}$. Then we have
\begin{equation}
\begin{aligned}
\rho^{Q_{1}^{N}}&=\rho^{Q_{1}} \otimes \cdots \otimes \rho^{Q_{N}}\\
&=(q\ket{0}\bra{0}+(1-q)\ket{1}\bra{1})^{\otimes N}\\
&=\sum_{Q_{1}^{N} \in \mathcal{X}^{N}}Pr\left(\ket{Q_{1}^{N}}\bra{Q_{1}^{N}}\right)\ket{Q_{1}^{N}}\bra{Q_{1}^{N}}
\end{aligned}
\end{equation}
where $Pr\left(\ket{Q_{1}^{N}}\bra{Q_{1}^{N}}\right)=\prod_{i=1}^{N}Pr(\ket{Q_{i}}\bra{Q_{i}})$, alphabet $\mathcal{X}=\{0,1\}$ and $\mathcal{X}^{N}$ is the N-power extension alphabet of $\mathcal{X}$. Introduce reference system $\rho^{R_{1}^{N}}=\rho^{R_{1}} \otimes \cdots \otimes \rho^{R_{N}}$ to purify $\rho^{Q_{i}^{N}}$, 
where $\rho^{R_{1}} = \cdots = \rho^{R_{N}}=\rho^{Q_{1}} = \cdots = \rho^{Q_{N}}$. We have
\begin{equation}
\ket{\varphi_{Q_{1}^{N},R_{1}^{N}}}=\sum_{Q_{1}^{N}=R_{1}^{N} \in \mathcal{X}^{N}} \sqrt{Pr\left(\ket{Q_{1}^{N}}\bra{Q_{1}^{N}}\right)}\ket{Q_{1}^{N},R_{1}^{N}}
\end{equation}
Then the density operator $\rho^{Q_{1}^{N},R_{1}^{N}}$ of the joint system $Q_{1}^{N},R_{1}^{N}$ is
\begin{equation}
\begin{aligned}
&\rho^{Q_{1}^{N},R_{1}^{N}}\\
&=\ket{\varphi_{Q_{1}^{N},R_{1}^{N}}}\bra{\varphi_{Q_{1}^{N},R_{1}^{N}}}\\
&=\sum_{\substack{Q_{1}^{N}=R_{1}^{N} \in \mathcal{X}^{N}\\ \tilde{Q}_{1}^{N}=\tilde{R}_{1}^{N} \in \mathcal{X}^{N}}} \sqrt{Pr\left(\ket{Q_{1}^{N}}\bra{Q_{1}^{N}}\right)}\ket{Q_{1}^{N},R_{1}^{N}} \sqrt{Pr\left(\ket{\tilde{Q}_{1}^{N}}\bra{\tilde{Q}_{1}^{N}}\right)}\bra{\tilde{Q}_{1}^{N},\tilde{R}_{1}^{N}}
\end{aligned}
\end{equation}
We use a unitary operator $U_N$ which only acts on system $Q_{1}^{N}$ to represent the encoding process $\ket{Q_{1}^{N}}\rightarrow\ket{C_{1}^{N}}$, and we have
\begin{equation}
\begin{aligned}
&\rho^{C_{1}^{N},R_{1}^{N}}\\
&=U_{N} \rho^{Q_{1}^{N} R_{1}^{N}} U_{N}^{\dagger}\\
&=U_{N} \left(\sum_{\substack{Q_{1}^{N}=R_{1}^{N} \in \mathcal{X}^{N}\\ \tilde{Q}_{1}^{N}=\tilde{R}_{1}^{N} \in \mathcal{X}^{N}}} \sqrt{Pr\left(\ket{Q_{1}^{N}}\bra{Q_{1}^{N}}\right)}\ket{Q_{1}^{N},R_{1}^{N}} \sqrt{Pr\left(\ket{\tilde{Q}_{1}^{N}}\bra{\tilde{Q}_{1}^{N}}\right)}\bra{\tilde{Q}_{1}^{N},\tilde{R}_{1}^{N}}\right)U_{N}^{\dagger}\\
&=\sum_{\substack{Q_{1}^{N}=R_{1}^{N} \in \mathcal{X}^{N}\\\tilde{Q}_{1}^{N}=\tilde{R}_{1}^{N} \in \mathcal{X}^{N}}} \sqrt{Pr\left(\ket{Q_{1}^{N}}\bra{Q_{1}^{N}}\right)}\ket{Q_{1}^{N}G_{N},R_{1}^{N}} \sqrt{Pr\left(\ket{\tilde{Q}_{1}^{N}}\bra{\tilde{Q}_{1}^{N}}\right)}\bra{\tilde{Q}_{1}^{N}G_{N},\tilde{R}_{1}^{N}}\\
&=\sum_{R_{1}^{N} \in \mathcal{X}^{N}} \sqrt{Pr\left(\ket{Q_{1}^{N}}\bra{Q_{1}^{N}}\right)}\ket{C_{1}^{N},R_{1}^{N}}\sum_{\tilde{R}_{1}^{N} \in \mathcal{X}^{N}} \sqrt{Pr\left(\ket{\tilde{Q}_{1}^{N}}\bra{\tilde{Q}_{1}^{N}}\right)}\bra{\tilde{C}_{1}^{N},\tilde{R}_{1}^{N}}
\end{aligned}
\end{equation}
where $C_{1}^{N}=Q_{1}^{N} G_{N}$, $\tilde{C}_{1}^{N}=\tilde{Q}_{1}^{N} G_{N}$ and $G_{N}$ is generator matrix.

The channel $\mathcal{E}^{\otimes N}$, whose operator elements are $\{ F_{k}\}_{k=0, \cdots, 2^{N}-1}$, follows the encoding process $\ket{Q_{1}^{N}}\rightarrow\ket{C_{1}^{N}}$. Then the density operator $\rho^{V_{1}^{N},R_{1}^{N}}$ of the output of the channel $\mathcal{E}^{\otimes N}$ is
\begin{equation}
\begin{aligned}
\rho^{V_{1}^{N},R_{1}^{N}}&=\sum_{k=0}^{2^{N}-1}F_{k} \rho^{C_{1}^{N},R_{1}^{N}} F_{k}^{\dagger} \\
&=\sum_{k=0}^{2^{N}-1}F_{k}\sum_{Q_{1}^{N}=R_{1}^{N} \in \mathcal{X}^{N}} \sqrt{Pr\left(\ket{Q_{1}^{N}}\bra{Q_{1}^{N}}\right)}\ket{Q_{1}^{N}G_{N},R_{1}^{N}}\\
&\times \sum_{\tilde{Q}_{1}^{N}=\tilde{R}_{1}^{N} \in \mathcal{X}^{N}} \sqrt{Pr\left(\ket{\tilde{Q}_{1}^{N}}\bra{\tilde{Q}_{1}^{N}}\right)}\bra{\tilde{Q}_{1}^{N}G_{N},\tilde{R}_{1}^{N}} F_{k}^{\dagger}
\end{aligned}
\end{equation}

Notice that the channel $\mathcal{E}^{\otimes N}$ is the last layer of the channel $\mathcal{E}_{N}$, so the density operator $\rho^{V_{1}^{N},R_{1}^{N}}$ is also the output of the channel $\mathcal{E}_{N}$. Then we perform partial trace over the system
$R_{i}^{N}$ and obtain
\begin{equation}
\begin{aligned}
\begin{split}
\label{rho of V1N,R1i-1}
&\rho^{V_{1}^{N},R_{1}^{N}}\\
&=tr_{R_{i}^{N}}\left[ \sum_{k=0}^{2^{N}-1}F_{k}\sum_{Q_{1}^{N}=R_{1}^{N} \in \mathcal{X}^{N}} \sqrt{Pr\left(\ket{Q_{1}^{N}}\bra{Q_{1}^{N}}\right)}\ket{Q_{1}^{N}G_{N},R_{1}^{N}}\right. \\
&\times \left.\sum_{\tilde{Q}_{1}^{N}=\tilde{R}_{1}^{N} \in \mathcal{X}^{N}} \sqrt{Pr\left(\ket{\tilde{Q}_{1}^{N}}\bra{\tilde{Q}_{1}^{N}}\right)}\bra{\tilde{Q}_{1}^{N}G_{N},\tilde{R}_{1}^{N}} F_{k}^{\dagger} \right]\\
&=tr_{R_{i}^{N}}\left[ \sum_{k=0}^{2^{N}-1}F_{k}\sum_{Q_{i}^{N}=R_{i}^{N}\in\mathcal{X}^{N-i+1}}\right.\\
&\times \left. \sum_{Q_{1}^{i-1}=R_{1}^{i-1} \in \mathcal{X}^{i-1}} \sqrt{Pr\left(\ket{Q_{1}^{i-1}}\bra{Q_{1}^{i-1}}\right)Pr\left(\ket{Q_{i}^{N}}\bra{Q_{i}^{N}}\right)}\ket{Q_{1}^{N}G_{N},R_{1}^{i-1},R_{i}^{N}}\right.\\
&\times \left.\sum_{\tilde{Q}_{1}^{i-1}=\tilde{R}_{1}^{i-1} \in \mathcal{X}^{i-1}} \sqrt{Pr\left(\ket{\tilde{Q}_{1}^{i-1}}\bra{\tilde{Q}_{1}^{i-1}}\right)Pr\left(\ket{Q_{i}^{N}}\bra{Q_{i}^{N}}\right)}\bra{\tilde{Q}_{1}^{N}G_{N},\tilde{R}_{1}^{i-1},R_{i}^{N}} F_{k}^{\dagger}\right]\\
&=\sum_{k=0}^{2^{N}-1}F_{k}\left[\sum_{Q_{i}^{N}\in\mathcal{X}^{N-i+1}}Pr\left(\ket{Q_{i}^{N}}\bra{Q_{i}^{N}}\right)\right.\\
&\times \left.\sum_{Q_{1}^{i-1}=R_{1}^{i-1} \in \mathcal{X}^{i-1}} \sqrt{Pr\left(\ket{Q_{1}^{i-1}}\bra{Q_{1}^{i-1}}\right)}\ket{Q_{1}^{N}G_{N},R_{1}^{i-1}}\right.\\
&\left.\times \sum_{\tilde{Q}_{1}^{i-1}=\tilde{R}_{1}^{i-1} \in \mathcal{X}^{i-1}} \sqrt{Pr\left(\ket{\tilde{Q}_{1}^{i-1}}\bra{\tilde{Q}_{1}^{i-1}}\right)}\bra{\tilde{Q}_{1}^{N}G_{N},\tilde{R}_{1}^{i-1}}\right] F_{k}^{\dagger}
\end{split}
\end{aligned}
\end{equation}
Eq.  (\ref{rho of V1N,R1i-1}) guarantees that
\begin{equation}
\begin{aligned}
\sum_{Q_{1}^{i-1}=R_{1}^{i-1} \in \mathcal{X}^{i-1}} \sqrt{Pr\left(\ket{Q_{1}^{i-1}}\bra{Q_{1}^{i-1}}\right)}\ket{Q_{1}^{N}G_{N},R_{1}^{i-1}}
\end{aligned}
\end{equation}
must be a unit vector, since it is easy to verify $\sum_{Q_{1}^{i-1}\in\mathcal{X}^{i-1}}Pr\left(\ket{Q_{1}^{i-1}}\bra{Q_{1}^{i-1}}\right)=1$. Divide the Eq.  (\ref{rho of V1N,R1i-1}) into two parts: $Q_{i}=0$ and $Q_{i}=1$, we have
\begin{equation}
\rho^{V_{1}^{N},R_{1}^{i-1}}=\rho_{V_{1}^{N},R_{1}^{i-1}}^{(0)}+\rho_{V_{1}^{N},R_{1}^{i-1}}^{(1)}
\end{equation}
where
\begin{equation}
\begin{aligned}
\begin{split}
\rho_{V_{1}^{N},R_{1}^{i-1}}^{(0)}&=q\sum_{k=0}^{2^{N}-1}F_{k}\left[\sum_{Q_{i+1}^{N}\in\mathcal{X}^{N-i}}Pr\left(\ket{Q_{i+1}^{N}}\bra{Q_{i+1}^{N}}\right)\right.\\
&\times \left.\sum_{Q_{1}^{i-1}=R_{1}^{i-1} \in \mathcal{X}^{i-1}} \sqrt{Pr\left(\ket{Q_{1}^{i-1}}\bra{Q_{1}^{i-1}}\right)}\ket{(Q_{1}^{i-1},0,Q_{i+1}^{N})G_{N},R_{1}^{i-1}}\right.\\
&\times \left. \sum_{\tilde{Q}_{1}^{i-1}=\tilde{R}_{1}^{i-1} \in \mathcal{X}^{i-1}} \sqrt{Pr\left(\ket{\tilde{Q}_{1}^{i-1}}\bra{\tilde{Q}_{1}^{i-1}}\right)}\bra{(\tilde{Q}_{1}^{i-1},0,Q_{i+1}^{N})G_{N},\tilde{R}_{1}^{i-1}}\right] F_{k}^{\dagger}
\end{split}
\end{aligned}
\end{equation}
and
\begin{equation}
\begin{aligned}
\begin{split}
\rho_{V_{1}^{N},R_{1}^{i-1}}^{(1)}&=(1-q)\sum_{k=0}^{2^{N}-1}F_{k}\left[\sum_{Q_{i+1}^{N}\in\mathcal{X}^{N-i}}Pr\left(\ket{Q_{i+1}^{N}}\bra{Q_{i+1}^{N}}\right)\right.\\
&\times \left.\sum_{Q_{1}^{i-1}=R_{1}^{i-1} \in \mathcal{X}^{i-1}} \sqrt{Pr\left(\ket{Q_{1}^{i-1}}\bra{Q_{1}^{i-1}}\right)}\ket{(Q_{1}^{i-1},1,Q_{i+1}^{N})G_{N},R_{1}^{i-1}}\right.\\
&\times \left. \sum_{\tilde{Q}_{1}^{i-1}=\tilde{R}_{1}^{i-1} \in \mathcal{X}^{i-1}} \sqrt{Pr\left(\ket{\tilde{Q}_{1}^{i-1}}\bra{\tilde{Q}_{1}^{i-1}}\right)}\bra{(\tilde{Q}_{1}^{i-1},1,Q_{i+1}^{N})G_{N},\tilde{R}_{1}^{i-1}}\right] F_{k}^{\dagger}
\end{split}
\end{aligned}
\end{equation}

For $\rho_{V_{1}^{N},R_{1}^{i-1}}^{(0)}$ and $\rho_{V_{1}^{N},R_{1}^{i-1}}^{(1)}$, we exchange summation order, and obtain
\begin{equation}
\begin{aligned}
\begin{split}
\label{rho0}
&\rho_{V_{1}^{N},R_{1}^{i-1}}^{(0)}\\
&=q\sum_{Q_{i+1}^{N}\in\mathcal{X}^{N-i}}Pr\left(\ket{Q_{i+1}^{N}}\bra{Q_{i+1}^{N}}\right)\\
&\times \sum_{k=0}^{2^{N}-1}F_{k}\left[\sum_{Q_{1}^{i-1}=R_{1}^{i-1} \in \mathcal{X}^{i-1}} \sqrt{Pr\left(\ket{Q_{1}^{i-1}}\bra{Q_{1}^{i-1}}\right)} \ket{(Q_{1}^{i-1},0,Q_{i+1}^{N})G_{N},R_{1}^{i-1}}\right.\\
&\left.\times \sum_{\tilde{Q}_{1}^{i-1}=\tilde{R}_{1}^{i-1} \in \mathcal{X}^{i-1}} \sqrt{Pr\left(\ket{\tilde{Q}_{1}^{i-1}}\bra{\tilde{Q}_{1}^{i-1}}\right)}\bra{(\tilde{Q}_{1}^{i-1},0,Q_{i+1}^{N})G_{N},\tilde{R}_{1}^{i-1}}\right] F_{k}^{\dagger}
\end{split}
\end{aligned}
\end{equation}
and
\begin{equation}
\begin{aligned}
\begin{split}
\label{rho1}
&\rho_{V_{1}^{N},R_{1}^{i-1}}^{(1)}\\
&=(1-q)\sum_{Q_{i+1}^{N}\in\mathcal{X}^{N-i}}Pr\left(\ket{Q_{i+1}^{N}}\bra{Q_{i+1}^{N}}\right)\\
&\times \sum_{k=0}^{2^{N}-1}F_{k}\left[\sum_{Q_{1}^{i-1}=R_{1}^{i-1} \in \mathcal{X}^{i-1}} \sqrt{Pr\left(\ket{Q_{1}^{i-1}}\bra{Q_{1}^{i-1}}\right)} \ket{(Q_{1}^{i-1},1,Q_{i+1}^{N})G_{N},R_{1}^{i-1}}\right.\\
&\left.\times \sum_{\tilde{Q}_{1}^{i-1}=\tilde{R}_{1}^{i-1} \in \mathcal{X}^{i-1}} \sqrt{Pr\left(\ket{\tilde{Q}_{1}^{i-1}}\bra{\tilde{Q}_{1}^{i-1}}\right)} \bra{(\tilde{Q}_{1}^{i-1},1,Q_{i+1}^{N})G_{N},\tilde{R}_{1}^{i-1}}\right] F_{k}^{\dagger}
\end{split}
\end{aligned}
\end{equation}

\textbf{2. The second step of the proof: prove that the density operator $\rho^{V_{1}^{N},R_{1}^{i-1}}$ can be diagonalized with respect to a set of basis $\{ \ket{m^{'}}\}_{m=0,\cdots,2^{N}-1}$}.

We will prove that density operators $\rho_{V_{1}^{N},R_{1}^{i-1}}^{(0)}$ and $\rho_{V_{1}^{N},R_{1}^{i-1}}^{(1)}$ can be diagonalized with respect to a same set of basis $\{ \ket{m^{'}}\}_{m=0,\cdots,2^{N}-1}$, namely,
\begin{equation}
\rho_{V_{1}^{N},R_{1}^{i-1}}^{(0)}=q\sum_{m=0}^{2^{N}-1}Pr_{N}^{(i)}\left(\ket{m^{'}}\lvert \ket{0}\right)\ket{m^{'}}\bra{m^{'}}
\end{equation}
\begin{equation}
\label{rho1}
\rho_{V_{1}^{N},R_{1}^{i-1}}^{(1)}=(1-q)\sum_{m=0}^{2^{N}-1}Pr_{N}^{(i)}\left(\ket{m^{'}}\lvert \ket{1}\right)\ket{m^{'}}\bra{m^{'}}
\end{equation}

We consider $\rho_{V_{1}^{N},R_{1}^{i-1}}^{(0)}$ only, since the proof method of $\rho_{V_{1}^{N},R_{1}^{i-1}}^{(1)}$ is the same as that of $\rho_{V_{1}^{N},R_{1}^{i-1}}^{(0)}$.
We first prove that the vector $\ket{m^{'}}$ can be written as
\begin{equation}
\begin{aligned}
\label{m'}
\ket{m^{'}}&=\sum_{Q_{1}^{i-1}=R_{1}^{i-1} \in \mathcal{X}^{i-1}} \sqrt{Pr\left(\ket{Q_{1}^{i-1}}\bra{Q_{1}^{i-1}}\right)}\ket{(Q_{1}^{i-1},0,0_{i+1}^{N})G_{N}\cdot V_{1}^{N},R_{1}^{i-1}}
\end{aligned}
\end{equation}

Since for all $Q_{i+1}^{N}\in \mathcal{X}^{N-i}$, operation elements  $\{F_{k}\}_{k=0,\cdots,2^{N}-1}$ will map vector
\begin{equation}
\begin{aligned}
\sum_{Q_{1}^{i-1}=R_{1}^{i-1}  \in \mathcal{X}^{i-1}} \sqrt{Pr\left(\ket{Q_{1}^{i-1}}\bra{Q_{1}^{i-1}}\right)}\ket{(Q_{1}^{i-1},0,Q_{i+1}^{N})G_{N},R_{1}^{i-1}}
\end{aligned}
\end{equation}
to a same set of orthogonal basis $\{ \ket{m^{'}}\}_{m=0,\cdots,2^{N}-1}$, which contains $2^{N}$ basis vectors. Thus, without losing generality, we can let $Q_{i+1}^{N}=0_{i+1}^{N}$. Using Eq.  (\ref{regulation 1}), Eq.  (\ref{regulation 2}) and Theorem \ref{theorem3-1:the quantum combined channel is a QSC}, we have
\begin{equation}
\begin{aligned}
\label{proof of m'}
&F_{k}\sum_{Q_{1}^{i-1}=R_{1}^{i-1} \in \mathcal{X}^{i-1}} \sqrt{Pr\left(\ket{Q_{1}^{i-1}}\bra{Q_{1}^{i-1}}\right)}\ket{(Q_{1}^{i-1},0,Q_{i+1}^{N})G_{N},R_{1}^{i-1}}\\
&=\sum_{Q_{1}^{i-1}=R_{1}^{i-1} \in \mathcal{X}^{i-1}}\sqrt{Pr_{N}\left(\ket{\left(Q_{1}^{i-1},0,0_{i+1}^{N}\right)G_{N}\cdot V_{1}^{N}}\lvert \ket{Q_{1}^{i-1},0,0_{i+1}^{N}}\right)} \\
&\times\sqrt{Pr\left(\ket{Q_{1}^{i-1}}\bra{Q_{1}^{i-1}}\right)}\ket{(Q_{1}^{i-1},0,0_{i+1}^{N})G_{N}\cdot V_{1}^{N},R_{1}^{i-1}}\\
&=\sqrt{Pr_{N}\left(\ket{V_{1}^{N}}\lvert \ket{0_{1}^{N}}\right)}\\
&\times\sum_{Q_{1}^{i-1}=R_{1}^{i-1} \in \mathcal{X}^{i-1}} \sqrt{Pr\left(\ket{Q_{1}^{i-1}}\bra{Q_{1}^{i-1}}\right)}\ket{(Q_{1}^{i-1},0,0_{i+1}^{N})G_{N}\cdot V_{1}^{N},R_{1}^{i-1}}\\
&=\sqrt{Pr_{N}\left(\ket{V_{1}^{N}}\lvert \ket{0_{1}^{N}}\right)}\ket{m^{'}}
\end{aligned}
\end{equation}
which proves the Eq.  (\ref{m'}).

Observe Eq.  (\ref{proof of m'}), there is a one-to-noe relationship between $F_{k}$ and $V_{1}^{N}$, thus sum over all $F_{k}$ is sum over all $V_{1}^{N}$ and Eq.  (\ref{rho0}) can be rewritten as
\begin{equation}
\begin{aligned}
\begin{split}
\label{new rho0}
&\rho_{V_{1}^{N},R_{1}^{i-1}}^{(0)}=q\sum_{Q_{i+1}^{N}\in\mathcal{X}^{N-i}}Pr\left(\ket{Q_{i+1}^{N}}\bra{Q_{i+1}^{N}}\right)\sum_{V_{1}^{N}\in \mathcal{Y}^{N}}\\
&\times \sum_{Q_{1}^{i-1}=R_{1}^{i-1} \in \mathcal{X}^{i-1}} Pr_{N}\left(\ket{\left(Q_{1}^{i-1},0,0_{i+1}^{N}\right)G_{N}\cdot V_{1}^{N}}\lvert \ket{Q_{1}^{i-1},0,Q_{i+1}^{N}}\right)\\
&\times \sqrt{Pr\left(\ket{Q_{1}^{i-1}}\bra{Q_{1}^{i-1}}\right)}\ket{(Q_{1}^{i-1},0,0_{i+1}^{N})G_{N}\cdot V_{1}^{N},R_{1}^{i-1}}\\
&\times \sum_{ \tilde{Q}_{1}^{i-1} =\tilde{R}_{1}^{i-1} \in \mathcal{X}^{i-1}} \sqrt{Pr\left(\ket{\tilde{Q}_{1}^{i-1}}\bra{\tilde{Q}_{1}^{i-1}}\right)}\bra{(\tilde{Q}_{1}^{i-1},0,0_{i+1}^{N})G_{N}\cdot V_{1}^{N},\tilde{R}_{1}^{i-1}}\\
&=q\sum_{Q_{i+1}^{N}\in\mathcal{X}^{N-i}}Pr\left(\ket{Q_{i+1}^{N}}\bra{Q_{i+1}^{N}}\right)\sum_{V_{1}^{N}\in \mathcal{Y}^{N}}Pr_{N}\left(\ket{V_{1}^{N}}\lvert \ket{0_{1}^{i-1},0,Q_{i+1}^{N}}\right)\\&\times \sum_{Q_{1}^{i-1}=R_{1}^{i-1} \in \mathcal{X}^{i-1}}\sqrt{Pr\left(\ket{Q_{1}^{i-1}}\bra{Q_{1}^{i-1}}\right)}\ket{(Q_{1}^{i-1},0,0_{i+1}^{N})G_{N}\cdot V_{1}^{N},R_{1}^{i-1}}\\
&\times \sum_{\tilde{Q}_{1}^{i-1}=\tilde{R}_{1}^{i-1} \in \mathcal{X}^{i-1}} \sqrt{Pr\left(\ket{\tilde{Q}_{1}^{i-1}}\bra{\tilde{Q}_{1}^{i-1}}\right)}\bra{(\tilde{Q}_{1}^{i-1},0,0_{i+1}^{N})G_{N}\cdot V_{1}^{N},\tilde{R}_{1}^{i-1}}
\end{split}
\end{aligned}
\end{equation}

Here we use the fact that $Pr_{N}\left(\ket{V_{1}^{N}}\lvert \ket{Q_{1}^{N}}\right)=Pr_{N}\left(\ket{a_{1}^{N} G_{N} \cdot V_{1}^{N}}\lvert \ket{Q_{1}^{N} \oplus a_{1}^{N}}\right)$ which is according to Theorem \ref{theorem3-1:the quantum combined channel is a QSC}, so let $a_{1}^{N}=Q_{1}^{i-1},0,0_{i+1}^{N}$, we have
\begin{equation}
\begin{aligned}
\label{theorem 3.1}
&Pr_{N}\left(\ket{V_{1}^{N}}\lvert \ket{0_{1}^{i-1},0,Q_{i+1}^{N}}\right)=Pr_{N}\left(\ket{\left(Q_{1}^{i-1},0,0_{i+1}^{N}\right)G_{N}\cdot V_{1}^{N}}\lvert \ket{Q_{1}^{i-1},0,Q_{i+1}^{N}}\right)
\end{aligned}
\end{equation}

For Eq.  (\ref{new rho0}), we exchange summation order and obtain
\begin{equation}
\begin{aligned}
\begin{split}
&\rho_{V_{1}^{N},R_{1}^{i-1}}^{(0)}\\
&=q\sum_{V_{1}^{N}\in\mathcal{Y}^{N}}\left[\sum_{Q_{i+1}^{N}\in\mathcal{X}^{N-i}}Pr\left(\ket{Q_{i+1}^{N}}\bra{Q_{i+1}^{N}}\right)Pr_{N}\left(\ket{V_{1}^{N}}\lvert \ket{0_{1}^{i-1},0,Q_{i+1}^{N}}\right)\right. \\
&\times \left. \sum_{Q_{1}^{i-1}=R_{1}^{i-1} \in \mathcal{X}^{i-1}}\sqrt{Pr\left(\ket{Q_{1}^{i-1}}\bra{Q_{1}^{i-1}}\right)} \ket{(Q_{1}^{i-1},0,0_{i+1}^{N})G_{N}\cdot V_{1}^{N},R_{1}^{i-1}}\right. \\
&\times \left.\sum_{\tilde{Q}_{1}^{i-1} =\tilde{R}_{1}^{i-1} \in \mathcal{X}^{i-1}} \sqrt{Pr\left(\ket{\tilde{Q}_{1}^{i-1}}\bra{\tilde{Q}_{1}^{i-1}}\right)} \bra{(\tilde{Q}_{1}^{i-1},0,0_{i+1}^{N})G_{N}\cdot V_{1}^{N},\tilde{R}_{1}^{i-1}}\right]\\
&=q\sum_{m=0}^{2^{N}-1}Pr_{N}^{(i)}\left(\ket{m^{'}}\lvert \ket{0}\right)\ket{m^{'}}\bra{m^{'}}
\end{split}
\end{aligned}
\end{equation}
where
\begin{equation}
\begin{aligned}
\ket{m^{'}}=\sum_{Q_{1}^{i-1}=R_{1}^{i-1} \in \mathcal{X}^{i-1}} \sqrt{Pr\left(\ket{Q_{1}^{i-1}}\bra{Q_{1}^{i-1}}\right)}\ket{(Q_{1}^{i-1},0,0_{i+1}^{N})G_{N}\cdot V_{1}^{N},R_{1}^{i-1}}
\end{aligned}
\end{equation}
and
\begin{equation}
\begin{aligned}
\label{PrNi0}
Pr_{N}^{(i)}\left(\ket{m^{'}}\lvert \ket{0}\right)=\sum_{Q_{i+1}^{N}\in \mathcal{X}^{N-i}}Pr\left(\ket{Q_{i+1}^{N}}\bra{Q_{i+1}^{N}}\right)Pr_{N}\left(\ket{V_{1}^{N}}\lvert \ket{0_{1}^{i-1},0,Q_{i+1}^{N}}\right)
\end{aligned}
\end{equation}

$Pr_{N}^{(i)}\left(\ket{m^{'}}\lvert \ket{0}\right)$ is the transition probability which means the probability of the input state $\ket{0}\bra{0}$ changing into $\ket{m^{'}}\bra{m^{'}}$. Using Eq.  (\ref{theorem 3.1}), then Eq.  (\ref{PrNi0}) can be rewritten as
\begin{equation}
\begin{aligned}
Pr_{N}^{(i)}\left(\ket{m^{'}}\lvert \ket{0}\right)&=\sum_{Q_{1}^{i-1}\in\mathcal{X}^{i-1}}\frac{1}{2^{i-1}}\sum_{Q_{i+1}^{N}\in \mathcal{X}^{N-i}}Pr\left(\ket{Q_{i+1}^{N}}\bra{Q_{i+1}^{N}}\right)\\
&\times Pr_{N}\left(\ket{\left(Q_{1}^{i-1},0,0_{i+1}^{N}\right)G_{N}\cdot V_{1}^{N}}\lvert \ket{Q_{1}^{i-1},0,Q_{i+1}^{N}}\right)
\end{aligned}
\end{equation}

Using the same method, Eq.  (\ref{rho1}) can be easily proved, and we have
\begin{equation}
\begin{aligned}
\label{PrNi1}
Pr_{N}^{(i)}\left(\ket{m^{'}}\lvert \ket{1}\right)&=\sum_{Q_{1}^{i-1}\in\mathcal{X}^{i-1}}\frac{1}{2^{i-1}}\sum_{Q_{i+1}^{N}\in \mathcal{X}^{N-i}}Pr\left(\ket{Q_{i+1}^{N}}\bra{Q_{i+1}^{N}}\right)\\
&\times Pr_{N}\left(\ket{\left(Q_{1}^{i-1},0,0_{i+1}^{N}\right)G_{N}\cdot V_{1}^{N}}\lvert\ket{Q_{1}^{i-1},1,Q_{i+1}^{N}}\right)\\
&=\sum_{Q_{i+1}^{N}\in \mathcal{X}^{N-i}}Pr\left(\ket{Q_{i+1}^{N}}\bra{Q_{i+1}^{N}}\right)Pr_{N}\left(\ket{V_{1}^{N}}\lvert\ket{0_{1}^{i-1},1,Q_{i+1}^{N}}\right)
\end{aligned}
\end{equation}

Thus, the basis transition probabilities can be uniformly expressed as
\begin{equation}
\begin{aligned}
\label{PrNi}
Pr_{N}^{(i)}\left(\ket{m^{'}}\lvert \ket{Q_{i}}\right)&=\sum_{Q_{1}^{i-1}\in\mathcal{X}^{i-1}}\frac{1}{2^{i-1}}\sum_{Q_{i+1}^{N}\in \mathcal{X}^{N-i}}Pr\left(\ket{Q_{i+1}^{N}}\bra{Q_{i+1}^{N}}\right)\\
&\times Pr_{N}\left(\ket{\left(Q_{1}^{i-1},0,0_{i+1}^{N}\right)G_{N}\cdot V_{1}^{N}}\lvert\ket{Q_{1}^{i-1},Q_{i},Q_{i+1}^{N}}\right)\\
&=\sum_{Q_{i+1}^{N} \in \mathcal{X}^{N-i}}Pr\left(\ket{Q_{i+1}^{N}}\bra{Q_{i+1}^{N}}\right)Pr_{N}\left(\ket{V_{1}^{N}}\lvert\ket{0_{1}^{i-1},Q_{i},Q_{i+1}^{N}}\right)
\end{aligned}
\end{equation}

\textbf{3. The third step of the proof: use Arikan’s method to prove the basis transition probability matrix is symmetric.}

Next, we will prove that the basis transition probability matrix is symmetric. We will refer to the proof method which Arikan used to prove that classical coordinate channels $\{W_{N}^{(i)}\}$ are symmetric if the primal binary-input discrete memoryless channel $W$ is symmetric.

By Theorem \ref{theorem3-1:the quantum combined channel is a QSC}, we have 
\begin{equation}
\begin{aligned}
&Pr_{N}\left(\ket{\left(Q_{1}^{i-1},0,0_{i+1}^{N}\right)G_{N}\cdot V_{1}^{N}}\lvert\ket{Q_{1}^{i-1},Q_{i},Q_{i+1}^{N}}\right)\\
&=Pr_{N}\left(\ket{(a_{1}^{i-1},1,a_{i+1}^{N})G_{N}\cdot(Q_{1}^{i-1},0,0_{i+1}^{N})G_{N}\cdot V_{1}^{N}}\lvert \right.\\
&\left. \ket{(Q_{1}^{i-1},Q_{i},Q_{i+1}^{N})\oplus (a_{1}^{i-1},1,a_{i+1}^{N})}\right)
\end{aligned}
\end{equation}
for arbitrary $(a_{1}^{i-1},1,a_{i+1}^{N})\in \mathcal{X}^{N}$, thus the Eq.  (\ref{PrNi}) can be rewritten as
\begin{equation}
\begin{aligned}
\label{new PrNi}
&Pr_{N}^{(i)}\left(\ket{m^{'}}\lvert \ket{Q_{i}}\right)\\
&=\sum_{Q_{1}^{i-1}\in\mathcal{X}^{i-1}}\frac{1}{2^{i-1}}\sum_{Q_{i+1}^{N}\in \mathcal{X}^{N-i}}Pr\left(\ket{Q_{i+1}^{N}}\bra{Q_{i+1}^{N}}\right)\\
&\times Pr_{N}\left(\ket{(Q_{1}^{i-1},0,0_{i+1}^{N})G_{N}\cdot V_{1}^{N}}\lvert \ket{Q_{1}^{i-1},Q_{i},Q_{i+1}^{N}}\right)\\
&=\sum_{Q_{1}^{i-1}\in\mathcal{X}^{i-1}}\frac{1}{2^{i-1}}\sum_{Q_{i+1}^{N}\in \mathcal{X}^{N-i}}Pr\left(\ket{Q_{i+1}^{N}}\bra{Q_{i+1}^{N}}\right)\\
&\times Pr_{N}\left(\ket{(a_{1}^{i-1},1,a_{i+1}^{N})G_{N}\cdot(Q_{1}^{i-1},0,0_{i+1}^{N})G_{N}\cdot V_{1}^{N}}\lvert \right.\\
&\left.\ket{(Q_{1}^{i-1},Q_{i},Q_{i+1}^{N})\oplus (a_{1}^{i-1},1,a_{i+1}^{N})}\right)\\
&=\sum_{Q_{1}^{i-1}\in\mathcal{X}^{i-1}}\frac{1}{2^{i-1}}\sum_{Q_{i+1}^{N}\in \mathcal{X}^{N-i}}Pr\left(\ket{Q_{i+1}^{N}}\bra{Q_{i+1}^{N}}\right)\\
&\times Pr_{N}\left(\ket{(a_{1}^{i-1},1,a_{i+1}^{N}\oplus Q_{1}^{i-1},0,0_{i+1}^{N})G_{N}\cdot V_{1}^{N}}\lvert\right.\\
&\left. \ket{(Q_{1}^{i-1},Q_{i},Q_{i+1}^{N})\oplus (a_{1}^{i-1},1,a_{i+1}^{N})}\right)\\
&=Pr_{N}^{(i)}\left(\sum_{Q_{1}^{i-1} =R_{1}^{i-1}\in\mathcal{X}^{i-1}}\sqrt{Pr\left(\ket{Q_{1}^{i-1}}\bra{Q_{1}^{i-1}}\right)}\right. \\
&\times \left. \ket{(a_{1}^{i-1},1,a_{i+1}^{N}\oplus Q_{1}^{i-1},0,0_{i+1}^{N})G_{N}\cdot V_{1}^{N},R_{1}^{i-1}\oplus a_{1}^{i-1}}\lvert\ket{Q_{i}\oplus 1} \right)
\end{aligned}
\end{equation}

Substitute Eq.  (\ref{m'}) into $Pr_{N}^{(i)}\left(\ket{m^{'}}\lvert \ket{Q_{i}}\right)$, and connect with Eq.  (\ref{new PrNi}), we have
\begin{equation}
\begin{aligned}
\label{PrNi is symmetric}
&Pr_{N}^{(i)}\left(\ket{m^{'}}\lvert \ket{Q_{i}}\right)\\
&=Pr_{N}^{(i)}\left(\sum_{\substack{Q_{1}^{i-1}\\ =R_{1}^{i-1}\\ \in \mathcal{X}^{i-1}}} \sqrt{Pr\left(\ket{Q_{1}^{i-1}}\bra{Q_{1}^{i-1}}\right)}\ket{\left(Q_{1}^{i-1}, 0,0_{i+1}^{N}\right) G_{N} \cdot V_{1}^{N}, R_{1}^{i-1}}\lvert \ket{Q_{i}}\right)\\
&=Pr_{N}^{(i)}\left(\sum_{Q_{1}^{i-1}=R_{1}^{i-1}\in\mathcal{X}^{i-1}}\sqrt{Pr\left(\ket{Q_{1}^{i-1}}\bra{Q_{1}^{i-1}}\right)}\right. \\
&\times \left.\ket{(a_{1}^{i-1},1,a_{i+1}^{N}\oplus Q_{1}^{i-1},0,0_{i+1}^{N})G_{N}\cdot V_{1}^{N},R_{1}^{i-1}\oplus a_{1}^{i-1}}\lvert\ket{Q_{i}\oplus 1} \right)
\end{aligned}
\end{equation}

Here we take $a_{1}^{N}= (a_{1}^{i-1},1,a_{i+1}^{N})$, and the proof is completed. The Eq.  (\ref{PrNi is symmetric}) means arbitrary row of the BTPM of the quantum coordinate channel $\mathcal{E}_{N}^{(i)}$ is a permutation of another row.
\end{proof}

\section {Proof of Proposition 9}
\label{Proof of Proposition 3}
In this section, we prove Proposition \ref{BTMP can be derived from TPM} that we can derive $Pr_{N}^{(i)}\left(\ket{m}\lvert \ket{Q_{i}}\right)$ from the TPM of classical coordinate channels $W_{N}^{(i)}$.
\begin{proof}
According to Arikan’s theorem\cite{5075875}, the transition probabilities of classical coordinate channels $\{W_{N}^{(i)}\}$ are 
\begin{equation}
\begin{aligned}
&W_{N}^{(i)}\left(y_{1}^{N}, u_{1}^{i-1} \mid u_{i}\right)\\
&=\sum_{u_{i+1}^{N} \in \mathcal{X}^{N-i}} \frac{1}{2^{N-1}} W_{N}\left(y_{1}^{N}\lvert u_{1}^{N}\right)\\
&=\sum_{u_{i+1}^{N} \in \mathcal{X}^{N-i}} \frac{1}{2^{N-1}} W_{N}\left( a_{1}^{N}G_{N}\cdot y_{1}^{N}\lvert u_{1}^{N}\oplus a_{1}^{N}\right)\\
&=W_{N}^{(i)}\left(a_{1}^{N}G_{N}\cdot y_{1}^{N}, u_{1}^{i-1}\oplus a_{1}^{i-1}\lvert u_{i}\oplus a_{i}\right)
\end{aligned}
\end{equation}
and Arikan has proved that classical combined channel $W_{N}$ and classical coordinate channels $\{W_{N}^{(i)}\}$ are
all symmetric, which satisfies
\begin{equation}
\begin{aligned}
\label{WN is symmetric}
&W_{N}\left(y_{1}^{N} \mid 0_{1}^{i-1},u_{i},u_{i+1}^{N}\right)=W_{N}\left(\left(u_{1}^{i-1},0,0_{i+1}^{N}\right) G_{N} \cdot y_{1}^{N} \mid u_{1}^{i-1},u_{i},u_{i+1}^{N}\right)
\end{aligned}
\end{equation}
and
\begin{equation}
\begin{aligned}
\label{WNi is symmetric}
&W_{N}^{(i)}\left(y_{1}^{N}, 0_{1}^{i-1} \mid u_{i}\right)=W_{N}^{(i)}\left(\left(u_{1}^{i-1},0,0_{i+1}^{N}\right) G_{N} \cdot y_{1}^{N}, 0_{1}^{i-1} \oplus u_{1}^{i-1} \mid u_{i}\right)
\end{aligned}
\end{equation}
for all $u_{1}^{i-1}\in \mathcal{X}^{i-1}$.

Using Theorem \ref{theorem3-1:the quantum combined channel is a QSC}, Proposition \ref{BTPM and TPM are the same} and Eq (\ref{WN is symmetric}), we have
\begin{equation}
\begin{aligned}
\label{PrN and WN}
Pr_{N}\left(\ket{V_{1}^{N}}\lvert \ket{0_{1}^{i-1}, Q_{i}, Q_{i+1}^{N}}\right)&=W_{N}\left(y_{1}^{N} \mid 0_{1}^{i-1},u_{i},u_{i+1}^{N}\right)\\
&=W_{N}\left(\left(u_{1}^{i-1},0,0_{i+1}^{N}\right) G_{N} \cdot y_{1}^{N} \mid u_{1}^{i-1},u_{i},u_{i+1}^{N}\right)
\\
&=Pr_{N}\left(\ket{(Q_{1}^{i-1}, 0, 0_{i+1}^{N}) G_{N}\cdot V_{1}^{N}}\lvert \ket{Q_{1}^{i-1}, Q_{i}, Q_{i+1}^{N}}\right)
\end{aligned}
\end{equation}
for all $V_{1}^{N}=y_{1}^{N} \in \mathcal{Y}^{N}$ and $Q_{1}^{N}=u_{1}^{N} \in \mathcal{X}^{N}$.

Substitute Eq. (\ref{PrN and WN}) and Eq. (\ref{WNi is symmetric}) into Eq. (\ref{new BTPM of rho of VIN,R1i-1}), we have
\begin{equation}
\begin{aligned}
\label{PrNi and WNi}
Pr_{N}^{(i)}\left(\ket{m}\lvert \ket{Q_{i}}\right)&=\sum_{u_{1}^{i-1}\in \mathcal{X}^{i-1}}\frac{1}{2^{N-1}}\\
&\times \sum_{u_{i+1}^{N}\in \mathcal{X}^{N-i}}W_{N}\left( (u_{1}^{i-1},0,0_{i+1}^{N})G_{N}\cdot y_{1}^{N}\lvert u_{1}^{i-1},u_{i},u_{i+1}^{N}\right)\\
&=\frac{2^{i-1}}{2^{N-1}}\sum_{u_{i+1}^{N}\in \mathcal{X}^{N-i}}W_{N}\left(y_{1}^{N}\lvert 0_{1}^{i-1},u_{i},u_{i+1}^{N}\right)\\
&=2^{i-1}W_{N}^{(i)}\left(y_{1}^{N}, 0_{1}^{i-1}\lvert u_{i}\right)
\end{aligned}
\end{equation}
The Eq. (\ref{PrNi and WNi}) means we can derive $Pr_{N}^{(i)}\left(\ket{m}\lvert \ket{Q_{i}}\right)$ from the TPM of classical coordinate channels $W_{N}^{(i)}$, which completes the proof.
\end{proof}

\end{appendices}





\bibliography{sn-bibliography}


\begin{thebibliography}{67}
\ifx \bisbn   \undefined \def \bisbn  #1{ISBN #1}\fi
\ifx \binits  \undefined \def \binits#1{#1}\fi
\ifx \bauthor  \undefined \def \bauthor#1{#1}\fi
\ifx \batitle  \undefined \def \batitle#1{#1}\fi
\ifx \bjtitle  \undefined \def \bjtitle#1{#1}\fi
\ifx \bvolume  \undefined \def \bvolume#1{\textbf{#1}}\fi
\ifx \byear  \undefined \def \byear#1{#1}\fi
\ifx \bissue  \undefined \def \bissue#1{#1}\fi
\ifx \bfpage  \undefined \def \bfpage#1{#1}\fi
\ifx \blpage  \undefined \def \blpage #1{#1}\fi
\ifx \burl  \undefined \def \burl#1{\textsf{#1}}\fi
\ifx \doiurl  \undefined \def \doiurl#1{\url{https://doi.org/#1}}\fi
\ifx \betal  \undefined \def \betal{\textit{et al.}}\fi
\ifx \binstitute  \undefined \def \binstitute#1{#1}\fi
\ifx \binstitutionaled  \undefined \def \binstitutionaled#1{#1}\fi
\ifx \bctitle  \undefined \def \bctitle#1{#1}\fi
\ifx \beditor  \undefined \def \beditor#1{#1}\fi
\ifx \bpublisher  \undefined \def \bpublisher#1{#1}\fi
\ifx \bbtitle  \undefined \def \bbtitle#1{#1}\fi
\ifx \bedition  \undefined \def \bedition#1{#1}\fi
\ifx \bseriesno  \undefined \def \bseriesno#1{#1}\fi
\ifx \blocation  \undefined \def \blocation#1{#1}\fi
\ifx \bsertitle  \undefined \def \bsertitle#1{#1}\fi
\ifx \bsnm \undefined \def \bsnm#1{#1}\fi
\ifx \bsuffix \undefined \def \bsuffix#1{#1}\fi
\ifx \bparticle \undefined \def \bparticle#1{#1}\fi
\ifx \barticle \undefined \def \barticle#1{#1}\fi
\bibcommenthead
\ifx \bconfdate \undefined \def \bconfdate #1{#1}\fi
\ifx \botherref \undefined \def \botherref #1{#1}\fi
\ifx \url \undefined \def \url#1{\textsf{#1}}\fi
\ifx \bchapter \undefined \def \bchapter#1{#1}\fi
\ifx \bbook \undefined \def \bbook#1{#1}\fi
\ifx \bcomment \undefined \def \bcomment#1{#1}\fi
\ifx \oauthor \undefined \def \oauthor#1{#1}\fi
\ifx \citeauthoryear \undefined \def \citeauthoryear#1{#1}\fi
\ifx \endbibitem  \undefined \def \endbibitem {}\fi
\ifx \bconflocation  \undefined \def \bconflocation#1{#1}\fi
\ifx \arxivurl  \undefined \def \arxivurl#1{\textsf{#1}}\fi
\csname PreBibitemsHook\endcsname

\bibitem{PhysRevA.52.R2493}
\begin{barticle}
\bauthor{\bsnm{Shor}, \binits{P.W.}}:
\batitle{Scheme for reducing decoherence in quantum computer memory}.
\bjtitle{Phys. Rev. A}
\bvolume{52},
\bfpage{2493}--\blpage{2496}
(\byear{1995}).
\doiurl{10.1103/PhysRevA.52.R2493}
\end{barticle}
\endbibitem

\bibitem{PhysRevLett.77.793}
\begin{barticle}
\bauthor{\bsnm{Steane}, \binits{A.M.}}:
\batitle{Error correcting codes in quantum theory}.
\bjtitle{Phys. Rev. Lett.}
\bvolume{77},
\bfpage{793}--\blpage{797}
(\byear{1996}).
\doiurl{10.1103/PhysRevLett.77.793}
\end{barticle}
\endbibitem

\bibitem{1057683}
\begin{barticle}
\bauthor{\bsnm{Gallager}, \binits{R.}}:
\batitle{Low-density parity-check codes}.
\bjtitle{IRE Transactions on Information Theory}
\bvolume{8}(\bissue{1}),
\bfpage{21}--\blpage{28}
(\byear{1962}).
\doiurl{10.1109/TIT.1962.1057683}
\end{barticle}
\endbibitem

\bibitem{748992}
\begin{barticle}
\bauthor{\bsnm{MacKay}, \binits{D.J.C.}}:
\batitle{Good error-correcting codes based on very sparse matrices}.
\bjtitle{IEEE Transactions on Information Theory}
\bvolume{45}(\bissue{2}),
\bfpage{399}--\blpage{431}
(\byear{1999}).
\doiurl{10.1109/18.748992}
\end{barticle}
\endbibitem

\bibitem{mackay1996near}
\begin{barticle}
\bauthor{\bsnm{MacKay}, \binits{D.J.}},
\bauthor{\bsnm{Neal}, \binits{R.M.}}:
\batitle{Near shannon limit performance of low density parity check codes}.
\bjtitle{Electronics letters}
\bvolume{32}(\bissue{18}),
\bfpage{1645}
(\byear{1996})
\end{barticle}
\endbibitem

\bibitem{5075875}
\begin{barticle}
\bauthor{\bsnm{Arikan}, \binits{E.}}:
\batitle{Channel polarization: A method for constructing capacity-achieving
  codes for symmetric binary-input memoryless channels}.
\bjtitle{IEEE Transactions on Information Theory}
\bvolume{55}(\bissue{7}),
\bfpage{3051}--\blpage{3073}
(\byear{2009}).
\doiurl{10.1109/TIT.2009.2021379}
\end{barticle}
\endbibitem

\bibitem{bravyi1998quantum}
\begin{botherref}
\oauthor{\bsnm{Bravyi}, \binits{S.B.}},
\oauthor{\bsnm{Kitaev}, \binits{A.Y.}}:
Quantum codes on a lattice with boundary.
arXiv preprint quant-ph/9811052
(1998)
\end{botherref}
\endbibitem

\bibitem{PhysRevA.89.022321}
\begin{barticle}
\bauthor{\bsnm{Stephens}, \binits{A.M.}}:
\batitle{Fault-tolerant thresholds for quantum error correction with the
  surface code}.
\bjtitle{Phys. Rev. A}
\bvolume{89},
\bfpage{022321}
(\byear{2014}).
\doiurl{10.1103/PhysRevA.89.022321}
\end{barticle}
\endbibitem

\bibitem{bullock2007qudit}
\begin{barticle}
\bauthor{\bsnm{Bullock}, \binits{S.S.}},
\bauthor{\bsnm{Brennen}, \binits{G.K.}}:
\batitle{Qudit surface codes and gauge theory with finite cyclic groups}.
\bjtitle{Journal of Physics A: Mathematical and Theoretical}
\bvolume{40}(\bissue{13}),
\bfpage{3481}
(\byear{2007})
\end{barticle}
\endbibitem

\bibitem{zemor2009cayley}
\begin{bchapter}
\bauthor{\bsnm{Z{\'e}mor}, \binits{G.}}:
\bctitle{On cayley graphs, surface codes, and the limits of homological coding
  for quantum error correction}.
In: \bbtitle{International Conference on Coding and Cryptology},
pp. \bfpage{259}--\blpage{273}
(\byear{2009}).
\bcomment{Springer}
\end{bchapter}
\endbibitem

\bibitem{wang2009threshold}
\begin{botherref}
\oauthor{\bsnm{Wang}, \binits{D.S.}},
\oauthor{\bsnm{Fowler}, \binits{A.G.}},
\oauthor{\bsnm{Stephens}, \binits{A.M.}},
\oauthor{\bsnm{Hollenberg}, \binits{L.C.L.}}:
Threshold error rates for the toric and surface codes.
arXiv preprint arXiv:0905.0531
(2009)
\end{botherref}
\endbibitem

\bibitem{PhysRevA.80.052312}
\begin{barticle}
\bauthor{\bsnm{Fowler}, \binits{A.G.}},
\bauthor{\bsnm{Stephens}, \binits{A.M.}},
\bauthor{\bsnm{Groszkowski}, \binits{P.}}:
\batitle{High-threshold universal quantum computation on the surface code}.
\bjtitle{Phys. Rev. A}
\bvolume{80},
\bfpage{052312}
(\byear{2009}).
\doiurl{10.1103/PhysRevA.80.052312}
\end{barticle}
\endbibitem

\bibitem{bravyi2012subsystem}
\begin{botherref}
\oauthor{\bsnm{Bravyi}, \binits{S.}},
\oauthor{\bsnm{Duclos-Cianci}, \binits{G.}},
\oauthor{\bsnm{Poulin}, \binits{D.}},
\oauthor{\bsnm{Suchara}, \binits{M.}}:
Subsystem surface codes with three-qubit check operators.
arXiv preprint arXiv:1207.1443
(2012)
\end{botherref}
\endbibitem

\bibitem{ghosh2012surface}
\begin{barticle}
\bauthor{\bsnm{Ghosh}, \binits{J.}},
\bauthor{\bsnm{Fowler}, \binits{A.G.}},
\bauthor{\bsnm{Geller}, \binits{M.R.}}:
\batitle{Surface code with decoherence: An analysis of three superconducting
  architectures}.
\bjtitle{Physical Review A}
\bvolume{86}(\bissue{6}),
\bfpage{062318}
(\byear{2012})
\end{barticle}
\endbibitem

\bibitem{PhysRevLett.109.180502}
\begin{barticle}
\bauthor{\bsnm{Fowler}, \binits{A.G.}}:
\batitle{Proof of finite surface code threshold for matching}.
\bjtitle{Phys. Rev. Lett.}
\bvolume{109},
\bfpage{180502}
(\byear{2012}).
\doiurl{10.1103/PhysRevLett.109.180502}
\end{barticle}
\endbibitem

\bibitem{PhysRevLett.109.160503}
\begin{barticle}
\bauthor{\bsnm{Wootton}, \binits{J.R.}},
\bauthor{\bsnm{Loss}, \binits{D.}}:
\batitle{High threshold error correction for the surface code}.
\bjtitle{Phys. Rev. Lett.}
\bvolume{109},
\bfpage{160503}
(\byear{2012}).
\doiurl{10.1103/PhysRevLett.109.160503}
\end{barticle}
\endbibitem

\bibitem{PhysRevA.86.042313}
\begin{barticle}
\bauthor{\bsnm{Fowler}, \binits{A.G.}},
\bauthor{\bsnm{Whiteside}, \binits{A.C.}},
\bauthor{\bsnm{Hollenberg}, \binits{L.C.L.}}:
\batitle{Towards practical classical processing for the surface code: Timing
  analysis}.
\bjtitle{Phys. Rev. A}
\bvolume{86},
\bfpage{042313}
(\byear{2012}).
\doiurl{10.1103/PhysRevA.86.042313}
\end{barticle}
\endbibitem

\bibitem{PhysRevA.86.032324}
\begin{barticle}
\bauthor{\bsnm{Fowler}, \binits{A.G.}},
\bauthor{\bsnm{Mariantoni}, \binits{M.}},
\bauthor{\bsnm{Martinis}, \binits{J.M.}},
\bauthor{\bsnm{Cleland}, \binits{A.N.}}:
\batitle{Surface codes: Towards practical large-scale quantum computation}.
\bjtitle{Phys. Rev. A}
\bvolume{86},
\bfpage{032324}
(\byear{2012}).
\doiurl{10.1103/PhysRevA.86.032324}
\end{barticle}
\endbibitem

\bibitem{fowler2013optimal}
\begin{botherref}
\oauthor{\bsnm{Fowler}, \binits{A.G.}}:
Optimal complexity correction of correlated errors in the surface code.
arXiv preprint arXiv:1310.0863
(2013)
\end{botherref}
\endbibitem

\bibitem{barends2014superconducting}
\begin{barticle}
\bauthor{\bsnm{Barends}, \binits{R.}},
\bauthor{\bsnm{Kelly}, \binits{J.}},
\bauthor{\bsnm{Megrant}, \binits{A.}},
\bauthor{\bsnm{Veitia}, \binits{A.}},
\bauthor{\bsnm{Sank}, \binits{D.}},
\bauthor{\bsnm{Jeffrey}, \binits{E.}},
\bauthor{\bsnm{White}, \binits{T.C.}},
\bauthor{\bsnm{Mutus}, \binits{J.}},
\bauthor{\bsnm{Fowler}, \binits{A.G.}},
\bauthor{\bsnm{Campbell}, \binits{B.}}, \betal:
\batitle{Superconducting quantum circuits at the surface code threshold for
  fault tolerance}.
\bjtitle{Nature}
\bvolume{508}(\bissue{7497}),
\bfpage{500}--\blpage{503}
(\byear{2014})
\end{barticle}
\endbibitem

\bibitem{hill2015surface}
\begin{barticle}
\bauthor{\bsnm{Hill}, \binits{C.D.}},
\bauthor{\bsnm{Peretz}, \binits{E.}},
\bauthor{\bsnm{Hile}, \binits{S.J.}},
\bauthor{\bsnm{House}, \binits{M.G.}},
\bauthor{\bsnm{Fuechsle}, \binits{M.}},
\bauthor{\bsnm{Rogge}, \binits{S.}},
\bauthor{\bsnm{Simmons}, \binits{M.Y.}},
\bauthor{\bsnm{Hollenberg}, \binits{L.C.}}:
\batitle{A surface code quantum computer in silicon}.
\bjtitle{Science advances}
\bvolume{1}(\bissue{9}),
\bfpage{1500707}
(\byear{2015})
\end{barticle}
\endbibitem

\bibitem{delfosse2016linear}
\begin{botherref}
\oauthor{\bsnm{Delfosse}, \binits{N.}},
\oauthor{\bsnm{Iyer}, \binits{P.}},
\oauthor{\bsnm{Poulin}, \binits{D.}}:
A linear-time benchmarking tool for generalized surface codes.
arXiv preprint arXiv:1611.04256
(2016)
\end{botherref}
\endbibitem

\bibitem{PhysRevApplied.8.034021}
\begin{barticle}
\bauthor{\bsnm{Versluis}, \binits{R.}},
\bauthor{\bsnm{Poletto}, \binits{S.}},
\bauthor{\bsnm{Khammassi}, \binits{N.}},
\bauthor{\bsnm{Tarasinski}, \binits{B.}},
\bauthor{\bsnm{Haider}, \binits{N.}},
\bauthor{\bsnm{Michalak}, \binits{D.J.}},
\bauthor{\bsnm{Bruno}, \binits{A.}},
\bauthor{\bsnm{Bertels}, \binits{K.}},
\bauthor{\bsnm{DiCarlo}, \binits{L.}}:
\batitle{Scalable quantum circuit and control for a superconducting surface
  code}.
\bjtitle{Phys. Rev. Applied}
\bvolume{8},
\bfpage{034021}
(\byear{2017}).
\doiurl{10.1103/PhysRevApplied.8.034021}
\end{barticle}
\endbibitem

\bibitem{huang2020alibaba}
\begin{botherref}
\oauthor{\bsnm{Huang}, \binits{C.}},
\oauthor{\bsnm{Ni}, \binits{X.}},
\oauthor{\bsnm{Zhang}, \binits{F.}},
\oauthor{\bsnm{Newman}, \binits{M.}},
\oauthor{\bsnm{Ding}, \binits{D.}},
\oauthor{\bsnm{Gao}, \binits{X.}},
\oauthor{\bsnm{Wang}, \binits{T.}},
\oauthor{\bsnm{Zhao}, \binits{H.-H.}},
\oauthor{\bsnm{Wu}, \binits{F.}},
\oauthor{\bsnm{Zhang}, \binits{G.}}, et al.:
Alibaba cloud quantum development platform: Surface code simulations with
  crosstalk.
arXiv preprint arXiv:2002.08918
(2020)
\end{botherref}
\endbibitem

\bibitem{doi:10.1137/S0097539799359385}
\begin{barticle}
\bauthor{\bsnm{Aharonov}, \binits{D.}},
\bauthor{\bsnm{Ben-Or}, \binits{M.}}:
\batitle{Fault-tolerant quantum computation with constant error rate}.
\bjtitle{SIAM Journal on Computing}
\bvolume{38}(\bissue{4}),
\bfpage{1207}--\blpage{1282}
(\byear{2008}).
\doiurl{10.1137/S0097539799359385}
\end{barticle}
\endbibitem

\bibitem{knill1996concatenated}
\begin{botherref}
\oauthor{\bsnm{Knill}, \binits{E.}},
\oauthor{\bsnm{Laflamme}, \binits{R.}}:
Concatenated quantum codes.
arXiv preprint quant-ph/9608012
(1996)
\end{botherref}
\endbibitem

\bibitem{knill2005quantum}
\begin{barticle}
\bauthor{\bsnm{Knill}, \binits{E.}}:
\batitle{Quantum computing with realistically noisy devices}.
\bjtitle{Nature}
\bvolume{434}(\bissue{7029}),
\bfpage{39}--\blpage{44}
(\byear{2005})
\end{barticle}
\endbibitem

\bibitem{gottesman2013fault}
\begin{botherref}
\oauthor{\bsnm{Gottesman}, \binits{D.}}:
Fault-tolerant quantum computation with constant overhead.
arXiv preprint arXiv:1310.2984
(2013)
\end{botherref}
\endbibitem

\bibitem{tillich2013quantum}
\begin{barticle}
\bauthor{\bsnm{Tillich}, \binits{J.-P.}},
\bauthor{\bsnm{Z{\'e}mor}, \binits{G.}}:
\batitle{Quantum ldpc codes with positive rate and minimum distance
  proportional to the square root of the blocklength}.
\bjtitle{IEEE Transactions on Information Theory}
\bvolume{60}(\bissue{2}),
\bfpage{1193}--\blpage{1202}
(\byear{2013})
\end{barticle}
\endbibitem

\bibitem{freedman2013quantum}
\begin{botherref}
\oauthor{\bsnm{Freedman}, \binits{M.H.}},
\oauthor{\bsnm{Hastings}, \binits{M.B.}}:
Quantum systems on non-$ k $-hyperfinite complexes: A generalization of
  classical statistical mechanics on expander graphs.
arXiv preprint arXiv:1301.1363
(2013)
\end{botherref}
\endbibitem

\bibitem{guth2014quantum}
\begin{barticle}
\bauthor{\bsnm{Guth}, \binits{L.}},
\bauthor{\bsnm{Lubotzky}, \binits{A.}}:
\batitle{Quantum error correcting codes and 4-dimensional arithmetic hyperbolic
  manifolds}.
\bjtitle{Journal of Mathematical Physics}
\bvolume{55}(\bissue{8}),
\bfpage{082202}
(\byear{2014})
\end{barticle}
\endbibitem

\bibitem{kovalev2013fault}
\begin{barticle}
\bauthor{\bsnm{Kovalev}, \binits{A.A.}},
\bauthor{\bsnm{Pryadko}, \binits{L.P.}}:
\batitle{Fault tolerance of quantum low-density parity check codes with
  sublinear distance scaling}.
\bjtitle{Physical Review A}
\bvolume{87}(\bissue{2}),
\bfpage{020304}
(\byear{2013})
\end{barticle}
\endbibitem

\bibitem{hastings2013decoding}
\begin{botherref}
\oauthor{\bsnm{Hastings}, \binits{M.B.}}:
Decoding in hyperbolic spaces: Ldpc codes with linear rate and efficient error
  correction.
arXiv preprint arXiv:1312.2546
(2013)
\end{botherref}
\endbibitem

\bibitem{breuckmann2016constructions}
\begin{barticle}
\bauthor{\bsnm{Breuckmann}, \binits{N.P.}},
\bauthor{\bsnm{Terhal}, \binits{B.M.}}:
\batitle{Constructions and noise threshold of hyperbolic surface codes}.
\bjtitle{IEEE transactions on Information Theory}
\bvolume{62}(\bissue{6}),
\bfpage{3731}--\blpage{3744}
(\byear{2016})
\end{barticle}
\endbibitem

\bibitem{breuckmann2017hyperbolic}
\begin{barticle}
\bauthor{\bsnm{Breuckmann}, \binits{N.P.}},
\bauthor{\bsnm{Vuillot}, \binits{C.}},
\bauthor{\bsnm{Campbell}, \binits{E.}},
\bauthor{\bsnm{Krishna}, \binits{A.}},
\bauthor{\bsnm{Terhal}, \binits{B.M.}}:
\batitle{Hyperbolic and semi-hyperbolic surface codes for quantum storage}.
\bjtitle{Quantum Science and Technology}
\bvolume{2}(\bissue{3}),
\bfpage{035007}
(\byear{2017})
\end{barticle}
\endbibitem

\bibitem{breuckmann2021single}
\begin{barticle}
\bauthor{\bsnm{Breuckmann}, \binits{N.P.}},
\bauthor{\bsnm{Londe}, \binits{V.}}:
\batitle{Single-shot decoding of linear rate ldpc quantum codes with high
  performance}.
\bjtitle{IEEE Transactions on Information Theory}
\bvolume{68}(\bissue{1}),
\bfpage{272}--\blpage{286}
(\byear{2021})
\end{barticle}
\endbibitem

\bibitem{grospellier2021combining}
\begin{barticle}
\bauthor{\bsnm{Grospellier}, \binits{A.}},
\bauthor{\bsnm{Grou{\`e}s}, \binits{L.}},
\bauthor{\bsnm{Krishna}, \binits{A.}},
\bauthor{\bsnm{Leverrier}, \binits{A.}}:
\batitle{Combining hard and soft decoders for hypergraph product codes}.
\bjtitle{Quantum}
\bvolume{5},
\bfpage{432}
(\byear{2021})
\end{barticle}
\endbibitem

\bibitem{guo2013polar}
\begin{barticle}
\bauthor{\bsnm{Guo}, \binits{Y.}},
\bauthor{\bsnm{Lee}, \binits{M.H.}},
\bauthor{\bsnm{Zeng}, \binits{G.}}:
\batitle{Polar quantum channel coding with optical multi-qubit entangling gates
  for capacity-achieving channels}.
\bjtitle{Quantum information processing}
\bvolume{12}(\bissue{4}),
\bfpage{1659}--\blpage{1676}
(\byear{2013})
\end{barticle}
\endbibitem

\bibitem{renes2012efficient}
\begin{barticle}
\bauthor{\bsnm{Renes}, \binits{J.M.}},
\bauthor{\bsnm{Dupuis}, \binits{F.}},
\bauthor{\bsnm{Renner}, \binits{R.}}:
\batitle{Efficient polar coding of quantum information}.
\bjtitle{Physical Review Letters}
\bvolume{109}(\bissue{5}),
\bfpage{050504}
(\byear{2012})
\end{barticle}
\endbibitem

\bibitem{wilde2013polar}
\begin{barticle}
\bauthor{\bsnm{Wilde}, \binits{M.M.}},
\bauthor{\bsnm{Guha}, \binits{S.}}:
\batitle{Polar codes for degradable quantum channels}.
\bjtitle{IEEE Transactions on Information Theory}
\bvolume{59}(\bissue{7}),
\bfpage{4718}--\blpage{4729}
(\byear{2013})
\end{barticle}
\endbibitem

\bibitem{hirche2015polar}
\begin{botherref}
\oauthor{\bsnm{Hirche}, \binits{C.}}:
Polar codes in quantum information theory.
arXiv preprint arXiv:1501.03737
(2015)
\end{botherref}
\endbibitem

\bibitem{7208851}
\begin{barticle}
\bauthor{\bsnm{Renes}, \binits{J.M.}},
\bauthor{\bsnm{Sutter}, \binits{D.}},
\bauthor{\bsnm{Dupuis}, \binits{F.}},
\bauthor{\bsnm{Renner}, \binits{R.}}:
\batitle{Efficient quantum polar codes requiring no preshared entanglement}.
\bjtitle{IEEE Transactions on Information Theory}
\bvolume{61}(\bissue{11}),
\bfpage{6395}--\blpage{6414}
(\byear{2015}).
\doiurl{10.1109/TIT.2015.2468084}
\end{barticle}
\endbibitem

\bibitem{7370934}
\begin{barticle}
\bauthor{\bsnm{Hirche}, \binits{C.}},
\bauthor{\bsnm{Morgan}, \binits{C.}},
\bauthor{\bsnm{Wilde}, \binits{M.M.}}:
\batitle{Polar codes in network quantum information theory}.
\bjtitle{IEEE Transactions on Information Theory}
\bvolume{62}(\bissue{2}),
\bfpage{915}--\blpage{924}
(\byear{2016}).
\doiurl{10.1109/TIT.2016.2514319}
\end{barticle}
\endbibitem

\bibitem{8989387}
\begin{bchapter}
\bauthor{\bsnm{Dupuis}, \binits{F.}},
\bauthor{\bsnm{Goswami}, \binits{A.}},
\bauthor{\bsnm{Mhalla}, \binits{M.}},
\bauthor{\bsnm{Savin}, \binits{V.}}:
\bctitle{Purely quantum polar codes}.
In: \bbtitle{2019 IEEE Information Theory Workshop (ITW)},
pp. \bfpage{1}--\blpage{5}
(\byear{2019}).
\doiurl{10.1109/ITW44776.2019.8989387}
\end{bchapter}
\endbibitem

\bibitem{8815775}
\begin{barticle}
\bauthor{\bsnm{Babar}, \binits{Z.}},
\bauthor{\bsnm{Kaykac~Egilmez}, \binits{Z.B.}},
\bauthor{\bsnm{Xiang}, \binits{L.}},
\bauthor{\bsnm{Chandra}, \binits{D.}},
\bauthor{\bsnm{Maunder}, \binits{R.G.}},
\bauthor{\bsnm{Ng}, \binits{S.X.}},
\bauthor{\bsnm{Hanzo}, \binits{L.}}:
\batitle{Polar codes and their quantum-domain counterparts}.
\bjtitle{IEEE Communications Surveys Tutorials}
\bvolume{22}(\bissue{1}),
\bfpage{123}--\blpage{155}
(\byear{2020}).
\doiurl{10.1109/COMST.2019.2937923}
\end{barticle}
\endbibitem

\bibitem{6302198}
\begin{barticle}
\bauthor{\bsnm{Wilde}, \binits{M.M.}},
\bauthor{\bsnm{Guha}, \binits{S.}}:
\batitle{Polar codes for classical-quantum channels}.
\bjtitle{IEEE Transactions on Information Theory}
\bvolume{59}(\bissue{2}),
\bfpage{1175}--\blpage{1187}
(\byear{2013}).
\doiurl{10.1109/TIT.2012.2218792}
\end{barticle}
\endbibitem

\bibitem{6284203}
\begin{bchapter}
\bauthor{\bsnm{Wilde}, \binits{M.M.}},
\bauthor{\bsnm{Renes}, \binits{J.M.}}:
\bctitle{Quantum polar codes for arbitrary channels}.
In: \bbtitle{2012 IEEE International Symposium on Information Theory
  Proceedings},
pp. \bfpage{334}--\blpage{338}
(\byear{2012}).
\doiurl{10.1109/ISIT.2012.6284203}
\end{bchapter}
\endbibitem

\bibitem{goswami2021quantum}
\begin{botherref}
\oauthor{\bsnm{Goswami}, \binits{A.}},
\oauthor{\bsnm{Mhalla}, \binits{M.}},
\oauthor{\bsnm{Savin}, \binits{V.}}:
Quantum polarization of qudit channels.
arXiv preprint arXiv:2101.10194
(2021)
\end{botherref}
\endbibitem

\bibitem{9241807}
\begin{barticle}
\bauthor{\bsnm{Ramakrishnan}, \binits{N.}},
\bauthor{\bsnm{Iten}, \binits{R.}},
\bauthor{\bsnm{Scholz}, \binits{V.B.}},
\bauthor{\bsnm{Berta}, \binits{M.}}:
\batitle{Computing quantum channel capacities}.
\bjtitle{IEEE Transactions on Information Theory}
\bvolume{67}(\bissue{2}),
\bfpage{946}--\blpage{960}
(\byear{2021}).
\doiurl{10.1109/TIT.2020.3034471}
\end{barticle}
\endbibitem

\bibitem{8242350}
\begin{barticle}
\bauthor{\bsnm{Gyongyosi}, \binits{L.}},
\bauthor{\bsnm{Imre}, \binits{S.}},
\bauthor{\bsnm{Nguyen}, \binits{H.V.}}:
\batitle{A survey on quantum channel capacities}.
\bjtitle{IEEE Communications Surveys Tutorials}
\bvolume{20}(\bissue{2}),
\bfpage{1149}--\blpage{1205}
(\byear{2018}).
\doiurl{10.1109/COMST.2017.2786748}
\end{barticle}
\endbibitem

\bibitem{holevo2020quantum}
\begin{barticle}
\bauthor{\bsnm{Holevo}, \binits{A.S.}}:
\batitle{Quantum channel capacities}.
\bjtitle{Quantum Electronics}
\bvolume{50}(\bissue{5}),
\bfpage{440}
(\byear{2020})
\end{barticle}
\endbibitem

\bibitem{5592851}
\begin{bchapter}
\bauthor{\bsnm{Smith}, \binits{G.}}:
\bctitle{Quantum channel capacities}.
In: \bbtitle{2010 IEEE Information Theory Workshop},
pp. \bfpage{1}--\blpage{5}
(\byear{2010}).
\doiurl{10.1109/CIG.2010.5592851}
\end{bchapter}
\endbibitem

\bibitem{holevo2010mutual}
\begin{barticle}
\bauthor{\bsnm{Holevo}, \binits{A.S.}},
\bauthor{\bsnm{Shirokov}, \binits{M.E.}}:
\batitle{Mutual and coherent information for infinite-dimensional quantum
  channels}.
\bjtitle{Problems of information transmission}
\bvolume{46}(\bissue{3}),
\bfpage{201}--\blpage{218}
(\byear{2010})
\end{barticle}
\endbibitem

\bibitem{bennett2004quantum}
\begin{barticle}
\bauthor{\bsnm{Bennett}, \binits{C.H.}},
\bauthor{\bsnm{Shor}, \binits{P.W.}}:
\batitle{Quantum channel capacities}.
\bjtitle{Science}
\bvolume{303}(\bissue{5665}),
\bfpage{1784}--\blpage{1787}
(\byear{2004})
\end{barticle}
\endbibitem

\bibitem{PhysRevA.57.4153}
\begin{barticle}
\bauthor{\bsnm{Barnum}, \binits{H.}},
\bauthor{\bsnm{Nielsen}, \binits{M.A.}},
\bauthor{\bsnm{Schumacher}, \binits{B.}}:
\batitle{Information transmission through a noisy quantum channel}.
\bjtitle{Phys. Rev. A}
\bvolume{57},
\bfpage{4153}--\blpage{4175}
(\byear{1998}).
\doiurl{10.1103/PhysRevA.57.4153}
\end{barticle}
\endbibitem

\bibitem{PhysRevA.55.1613}
\begin{barticle}
\bauthor{\bsnm{Lloyd}, \binits{S.}}:
\batitle{Capacity of the noisy quantum channel}.
\bjtitle{Phys. Rev. A}
\bvolume{55},
\bfpage{1613}--\blpage{1622}
(\byear{1997}).
\doiurl{10.1103/PhysRevA.55.1613}
\end{barticle}
\endbibitem

\bibitem{kretschmann2004tema}
\begin{barticle}
\bauthor{\bsnm{Kretschmann}, \binits{D.}},
\bauthor{\bsnm{Werner}, \binits{R.F.}}:
\batitle{Tema con variazioni: quantum channel capacity}.
\bjtitle{New Journal of Physics}
\bvolume{6}(\bissue{1}),
\bfpage{26}
(\byear{2004})
\end{barticle}
\endbibitem

\bibitem{shor2003capacities}
\begin{botherref}
\oauthor{\bsnm{Shor}, \binits{P.W.}}:
Capacities of quantum channels and how to find them.
arXiv preprint quant-ph/0304102
(2003)
\end{botherref}
\endbibitem

\bibitem{javidian2021quantum}
\begin{botherref}
\oauthor{\bsnm{Javidian}, \binits{M.A.}},
\oauthor{\bsnm{Aggarwal}, \binits{V.}},
\oauthor{\bsnm{Bao}, \binits{F.}},
\oauthor{\bsnm{Jacob}, \binits{Z.}}:
Quantum entropic causal inference.
arXiv preprint arXiv:2102.11764
(2021)
\end{botherref}
\endbibitem

\bibitem{PhysRevA.54.2629}
\begin{barticle}
\bauthor{\bsnm{Schumacher}, \binits{B.}},
\bauthor{\bsnm{Nielsen}, \binits{M.A.}}:
\batitle{Quantum data processing and error correction}.
\bjtitle{Phys. Rev. A}
\bvolume{54},
\bfpage{2629}--\blpage{2635}
(\byear{1996}).
\doiurl{10.1103/PhysRevA.54.2629}
\end{barticle}
\endbibitem

\bibitem{nielsen2002quantum}
\begin{botherref}
\oauthor{\bsnm{Nielsen}, \binits{M.A.}},
\oauthor{\bsnm{Chuang}, \binits{I.}}:
Quantum computation and quantum information.
American Association of Physics Teachers
(2002)
\end{botherref}
\endbibitem

\bibitem{PhysRevA.54.2614}
\begin{barticle}
\bauthor{\bsnm{Schumacher}, \binits{B.}}:
\batitle{Sending entanglement through noisy quantum channels}.
\bjtitle{Phys. Rev. A}
\bvolume{54},
\bfpage{2614}--\blpage{2628}
(\byear{1996}).
\doiurl{10.1103/PhysRevA.54.2614}
\end{barticle}
\endbibitem

\bibitem{hastings2009superadditivity}
\begin{barticle}
\bauthor{\bsnm{Hastings}, \binits{M.B.}}:
\batitle{Superadditivity of communication capacity using entangled inputs}.
\bjtitle{Nature Physics}
\bvolume{5}(\bissue{4}),
\bfpage{255}--\blpage{257}
(\byear{2009})
\end{barticle}
\endbibitem

\bibitem{cubitt2015unbounded}
\begin{barticle}
\bauthor{\bsnm{Cubitt}, \binits{T.}},
\bauthor{\bsnm{Elkouss}, \binits{D.}},
\bauthor{\bsnm{Matthews}, \binits{W.}},
\bauthor{\bsnm{Ozols}, \binits{M.}},
\bauthor{\bsnm{P{\'e}rez-Garc{\'\i}a}, \binits{D.}},
\bauthor{\bsnm{Strelchuk}, \binits{S.}}:
\batitle{Unbounded number of channel uses may be required to detect quantum
  capacity}.
\bjtitle{Nature communications}
\bvolume{6}(\bissue{1}),
\bfpage{1}--\blpage{4}
(\byear{2015})
\end{barticle}
\endbibitem

\bibitem{smith2008quantum}
\begin{barticle}
\bauthor{\bsnm{Smith}, \binits{G.}},
\bauthor{\bsnm{Yard}, \binits{J.}}:
\batitle{Quantum communication with zero-capacity channels}.
\bjtitle{Science}
\bvolume{321}(\bissue{5897}),
\bfpage{1812}--\blpage{1815}
(\byear{2008})
\end{barticle}
\endbibitem

\bibitem{PhysRevLett.113.140401}
\begin{barticle}
\bauthor{\bsnm{Baumgratz}, \binits{T.}},
\bauthor{\bsnm{Cramer}, \binits{M.}},
\bauthor{\bsnm{Plenio}, \binits{M.B.}}:
\batitle{Quantifying coherence}.
\bjtitle{Phys. Rev. Lett.}
\bvolume{113},
\bfpage{140401}
(\byear{2014}).
\doiurl{10.1103/PhysRevLett.113.140401}
\end{barticle}
\endbibitem

\bibitem{PhysRevLett.115.020403}
\begin{barticle}
\bauthor{\bsnm{Streltsov}, \binits{A.}},
\bauthor{\bsnm{Singh}, \binits{U.}},
\bauthor{\bsnm{Dhar}, \binits{H.S.}},
\bauthor{\bsnm{Bera}, \binits{M.N.}},
\bauthor{\bsnm{Adesso}, \binits{G.}}:
\batitle{Measuring quantum coherence with entanglement}.
\bjtitle{Phys. Rev. Lett.}
\bvolume{115},
\bfpage{020403}
(\byear{2015}).
\doiurl{10.1103/PhysRevLett.115.020403}
\end{barticle}
\endbibitem

\end{thebibliography}


\end{document}